\documentclass [sigconf , nonacm]{acmart}
\pdfoutput=1
\usepackage{mathtools}
\usepackage{amsthm}
\usepackage{amsfonts}
\usepackage{nccmath}
\usepackage{amsmath}
\DeclareMathOperator*{\argmax}{arg\,max}

\usepackage[noend]{algorithmic}
\usepackage{graphicx}
\usepackage{float}
\usepackage{subcaption}
\usepackage{textcomp}
\usepackage{xcolor}
\usepackage{float}
\usepackage{setspace}
\usepackage{verbatim}
\usepackage[misc]{ifsym}
\usepackage{enumitem}
\usepackage{caption}
\usepackage{booktabs}
\usepackage{multirow}
\usepackage{balance}
\usepackage[a-2b]{pdfx}
\allowdisplaybreaks[4]

\usepackage{booktabs} 
\usepackage[ruled]{algorithm2e} 

\usepackage{enumitem}
\usepackage{dcolumn}

\usepackage{xcolor}

\newif\ifproofread
\newif\ifextendedreport
\makeatletter
\newcommand{\removelatexerror}{\let\@latex@error\@gobble}
\makeatother

\newcommand{\changemarker}[1]{%
	\ifproofread
	\textcolor{blue}{#1}%
	\else
	#1%
	\fi
}

\newcommand{\showappendix}[1]{%
	\ifextendedreport
	\clearpage
	\appendix
	#1%
	\else
	\fi
}

\newcommand{\venueforappendix}[1]{%
	\ifextendedreport
	Appendix #1%
	\else
	our extended report \cite{zheng2022secure}%
	\fi
}

\newcommand\vldbdoi{10.14778/3587136.3587141}
\newcommand\vldbpages{1657 - 1670}
\newcommand\vldbvolume{16}
\newcommand\vldbissue{7}
\newcommand\vldbyear{2023}
\newcommand\vldbauthors{\authors}
\newcommand\vldbtitle{\shorttitle} 
\newcommand\vldbavailabilityurl{https://github.com/teijyogen/SecSV}
\newcommand\vldbpagestyle{empty} 

\begin{document}
	\proofreadfalse
	\extendedreporttrue
	\title {Secure Shapley Value for Cross-Silo Federated Learning}
	
	\author{Shuyuan Zheng}
	\affiliation{%
		\institution{Kyoto University}
	}
	\email{caryzheng@db.soc.i.kyoto-u.ac.jp}

	\author{Yang Cao}
	\affiliation{%
		\institution{Hokkaido University}
	}
	\email{yang@ist.hokudai.ac.jp}
	
	\author{Masatoshi Yoshikawa}
	\affiliation{%
		\institution{Kyoto University}
	}
	\email{yoshikawa@i.kyoto-u.ac.jp}

	\begin{abstract}
		The Shapley value (SV) is a fair and principled metric for contribution evaluation in cross-silo federated learning (cross-silo FL), wherein organizations, i.e., clients, collaboratively train prediction models with the coordination of a parameter server.
		However, existing SV calculation methods for FL assume that the server can access the raw FL models and public test data.
		This may not be a valid assumption in practice considering the emerging privacy attacks on FL models and the fact that test data might be clients' private assets.  
		Hence, we investigate the problem of \textit{secure SV calculation} for cross-silo FL. 
		We first propose \textit{HESV}, a one-server solution based solely on homomorphic encryption (HE) for privacy protection, which has limitations in efficiency.
		To overcome these limitations, we propose \textit{SecSV}, an efficient two-server protocol with the following novel features.
		First, SecSV utilizes a hybrid privacy protection scheme to avoid ciphertext--ciphertext multiplications between test data and models, which are extremely expensive under HE.
		Second, an efficient secure matrix multiplication method is proposed for SecSV.
		Third, SecSV strategically identifies and skips some test samples without significantly affecting the evaluation accuracy.
		Our experiments demonstrate that SecSV is $7.2$-$36.6\times$ as fast as HESV, with a limited loss in the accuracy of calculated SVs.
	\end{abstract}
	
	\maketitle
	
	\pagestyle{\vldbpagestyle}
	\begingroup\small\noindent\raggedright\textbf{PVLDB Reference Format:}\\
	\vldbauthors. \vldbtitle. PVLDB, \vldbvolume(\vldbissue): \vldbpages, \vldbyear.\\
	\href{https://doi.org/\vldbdoi}{doi:\vldbdoi}
	\endgroup
	\begingroup
	\renewcommand\thefootnote{}\footnote{\noindent
		This work is licensed under the Creative Commons BY-NC-ND 4.0 International License. Visit \url{https://creativecommons.org/licenses/by-nc-nd/4.0/} to view a copy of this license. For any use beyond those covered by this license, obtain permission by emailing \href{mailto:info@vldb.org}{info@vldb.org}. Copyright is held by the owner/author(s). Publication rights licensed to the VLDB Endowment. \\
		\raggedright Proceedings of the VLDB Endowment, Vol. \vldbvolume, No. \vldbissue\ %
		ISSN 2150-8097. \\
		\href{https://doi.org/\vldbdoi}{doi:\vldbdoi} \\
	}\addtocounter{footnote}{-1}\endgroup
	
	\ifdefempty{\vldbavailabilityurl}{}{
		\vspace{.3cm}
		\begingroup\small\noindent\raggedright\textbf{PVLDB Artifact Availability:}\\
		The source code, data, and/or other artifacts have been made available at \url{\vldbavailabilityurl}.
		\endgroup
	}
	
	\section{Introduction}
	Personal data are perceived as the new oil of the data intelligence era.
	Organizations (e.g., banks and hospitals) can use machine learning (ML) on personal data to acquire valuable knowledge and intelligence to facilitate improved predictions and decisions.
	However, acquiring sufficient personal data for ML-based data analytics is often difficult due to numerous practical reasons, such as the limited user scale and diversity; organizations face considerable risk of privacy breaches by sharing user data. 
	Consequently, large personal data are stored as \textit{data silos} with few opportunities to extract the valuable information contained therein.
	
	To exploit the data silos in a privacy-preserving manner, \textit{cross-silo federated learning} (cross-silo FL, may be referred to as \textit{cross-organization FL}) \cite{yang2019federated, li2020federated, kairouz2021advances} was introduced as a promising paradigm for collaborative ML.
	It enables organizations, i.e., clients, to train an ML model without sharing user data, thus largely protecting privacy.
	Concretely, in a typical model-training process of FL, each client trains a \textit{local model} on her local side and uploads it to a server. The server then aggregates all the local models into a \textit{global model}, which contains knowledge learned from clients' data silos.
	However, recent studies have shown that sharing the local models or local updates may reveal private information \cite{zhu2019deep, zhao2020idlg, geiping2020inverting, yin2021see}. 
	Thus, some \textit{secure federated training} systems \cite{phong2018privacy, truex2019hybrid, zhang2020batchcrypt, sav2020poseidon, zhang2021dubhe, jiang2021flashe, ma2022privacy} have deployed \textit{homomorphic encryption} (HE) to prevent the server from accessing the raw models.
	As shown in Figure \ref{fig:sec_fed_train}, the clients encrypt local models using HE, and the server aggregates the encrypted local models to obtain an encrypted global model that can only be decrypted by the clients.
	
	In typical cases of cross-silo FL, a small number of organizations (e.g., banks and hospitals) collaboratively train an ML model for their own use. Their data may substantially vary in size, quality, and distribution, making their contributions to the model disparate.  
	Therefore, compensations are required to incentivize clients with high contributions to cooperate. 
	In such cases, the Shapley value (SV) \cite{shapley1953sv} is crucial to promoting fair cooperation, which is widely adopted as a fair and principled metric of contribution evaluation. 
	The SV calculates the average impact of a client's data on every possible subset of other clients' data as her contribution and can be used for many downstream tasks in FL or collaborative ML, such as data valuation \cite{ghorbani2019data, jia2019efficient, jia2019towards, wang2020principled, wei2020efficient, liu2021dealer}, revenue allocation \cite{song2019profit,ohrimenko2019collaborative,liu2020fedcoin, han2020replication}, reward discrimination \cite{sim2020collaborative, tay2022incentivizing}, and client selection \cite{nagalapatti2021game}.
	However, existing studies on SV calculation for FL \cite{song2019profit, wang2020principled, wei2020efficient, liu2022gtg} assume that the server can access the raw local models and public test data.
	This may not be a valid assumption given the emerging privacy attacks on local models \cite{zhu2019deep, zhao2020idlg, geiping2020inverting, yin2021see} and that in practice, test data may be private \cite{ohrimenko2019collaborative, fallah2020personalized, huang2021personalized, wang2019federated, paulik2021federated}.
	
	\begin{figure*}[ht]
		\centering
		\begin{minipage}[t]{0.49\linewidth}
			\centering
			\includegraphics[scale=0.5]{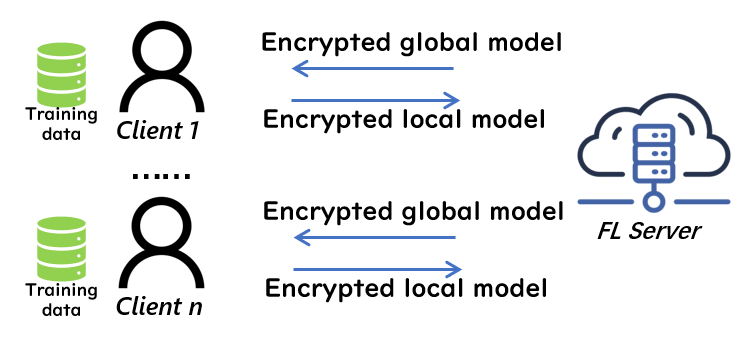}
			\subcaption{Secure federated training.}
			\label{fig:sec_fed_train}
		\end{minipage}
		\begin{minipage}[t]{0.49\linewidth}
			\centering
			\includegraphics[scale=0.5]{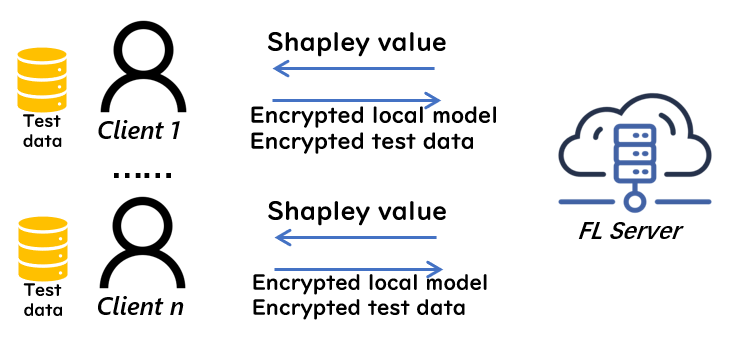}
			\subcaption{Secure SV calculation.}
			\label{fig:sec_sv_calc}
		\end{minipage}
		\vspace{-5pt}
		\caption{Secure federated training and SV calculation for cross-silo federated learning.}
	\end{figure*}

	\begin{example}
		Consider that some hospitals with different geographical distributions of patients collaboratively train an ML model for disease diagnosis by FL.
		Each hospital provides a training set of its patient data for federated training and a test set for SV-based contribution evaluation.
		As shown in Figures \ref{fig:sec_fed_train} and \ref{fig:sec_sv_calc}, they should ensure the security of both the training and SV calculation phases because the patient data are sensitive.
	\end{example}

	This study is the first to address \textit{secure Shapley value calculation} in cross-silo FL.
	Specifically, as depicted in Figure \ref{fig:sec_sv_calc}, extending the secure federated training for cross-silo FL \cite{phong2018privacy, truex2019hybrid, zhang2020batchcrypt, sav2020poseidon, zhang2021dubhe, jiang2021flashe, ma2022privacy}, we calculate SVs during the evaluation phase using homomorphically encrypted models.
	Solving this problem is significantly challenging because of the following two characteristics.
	First, secure SV calculation in cross-silo FL involves protecting both local models and client-side test data.
	Although the existing studies on SV calculation assume that a public test set is provided, in many practical cases (e.g., \cite{ohrimenko2019collaborative, fallah2020personalized, huang2021personalized, wang2019federated, paulik2021federated}), the test data are owned by clients as private assets (see Section \ref{sec:problem} for more details).
	Consequently, we protect  both models and test data in our problem, which cannot be supported by existing secure federated training methods because they only protect the model aggregation process.
	Second, SV calculation is an NP-hard task; it is computationally prohibitive to calculate SVs securely.
	To calculate SVs in FL, we must aggregate each subset of local models into an \textit{aggregated model} and evaluate its accuracy. This causes $O(2^{n})$ models to be tested, where $n$ is the number of clients.
	Although $n$ is relatively small in the cross-silo setting, secure SV calculation is highly inefficient because even testing only a single model is computationally expensive under HE.
	Moreover, the existing methods for accelerating SV calculation
	\cite{fatima2008linear, maleki2015addressing, ghorbani2019data, jia2019towards} can hardly improve the efficiency in our setting (as we will show in our experiments) because their strategies are based on sampling and testing a subset of the models to approximate the exact SVs, which is more suitable for large values of $n$.
	
	\subsection{Contributions}
	As a first step, we propose a one-server protocol for secure SV calculation, named \textit{HESV}.
	HESV only requires one server to calculate SVs and deploys a purely HE-based scheme for secure model testing, which suggests that both models and test data are encrypted using HE.
	However, the state-of-the-art (SOTA) homomorphic matrix multiplication method \cite{jiang2018e2dm} cannot multiply an encrypted matrix of model parameters by an encrypted matrix of input features of a batch of test samples when their sizes are large, which makes evaluating high-dimensional ML models infeasible.
	Hence, we propose an extended version of the SOTA, named \textit{Matrix Squaring}, which facilitates testing a wider range of models under HESV.
	Nevertheless, HESV has considerable limitations in efficiency:
	it involves computationally expensive \textit{ciphertext--ciphertext} (c2c) multiplications between encrypted models and data;
	it cannot encrypt together test samples from different clients in a single ciphertext to accelerate matrix multiplications;
	it evaluates the entire test set for all aggregated models, which is time-consuming.
	
	Subsequently, we propose \textit{SecSV}, a two-server protocol that overcomes the limitations of HESV. 
	SecSV is considerably more efficient than HESV, despite requiring an additional server to assist in secure SV calculation. This can be attributed to the following features. 
	First, we design a hybrid protection scheme for SecSV: models are encrypted using HE, whereas test data are protected by \textit{additive secret sharing} (ASS) \cite{shamir1979ass}. 
	Since test data protected by ASS are in plaintext, SecSV only calls for \textit{ciphertext--plaintext} (c2p) multiplications for model testing, which are significantly more efficient than c2c multiplications.
	Second, owing to the use of ASS, this hybrid scheme enables packing together and simultaneously evaluating numerous test samples from different clients to enhance efficiency. 
	Leveraging this feature, we propose a matrix multiplication method, \textit{Matrix Reducing}, which is significantly more efficient than Matrix Squaring when numerous test samples are batched together (i.e., the case under SecSV).
	Third, we propose an acceleration method for secure SV calculation called \textit{SampleSkip}.
	Our intuition is simple yet powerful:
	if some test samples can be correctly predicted by two models, these samples are likely to be correctly predicted by their aggregated model and thus skippable; 
	otherwise, these samples discriminate the models' utilities and clients' contributions.
	As test data are stored as plaintext under ASS, we can freely drop skippable samples from the batch for the aggregated model to considerably improve efficiency.
	Whereas existing SV acceleration methods \cite{fatima2008linear, maleki2015addressing, ghorbani2019data, jia2019towards} can hardly reduce the scale of model evaluations in cross-silo FL, SampleSkip always skips massive samples for testing and can be combined with these methods.
	
	Finally, we extensively verify the efficiency and effectiveness of our proposed techniques on diverse ML tasks such as \textit{image recognition}, \textit{news classification}, \textit{bank marketing}, and \textit{miRNA targeting}.
	SecSV achieves $7.2$-$36.6\times$ speedup w.r.t. HESV in our experiments.
	
	\section{Preliminary}
	\label{sec:pre}
	\subsubsection*{Machine learning task}
	In this paper, we focus on \textit{classification}, which is an ML task commonly considered in FL and covers a wide range of real-world applications.
	A $c$-class classifier is a prediction function $f: R^d \rightarrow R^c$ that given a $d$-sized \textit{feature vector} $x$ of a data sample yields a $c$-sized \textit{score vector} $f(x)$; 
	the argument $\hat{y} = \argmax_{j} f(x)[j]$ is assigned as the predicted label for features $x$, where $f(x)[j]$ denotes the $j$-th entry of $f(x)$.
	When a batch of $m$ samples is given, we overload the prediction function $f$ for classifying them:
	for a $d \times m$ \textit{feature matrix} $X$, where the $k$-th column is the feature vector of the $k$-th sample, the output $f(X)$ is a $c \times m$ \textit{score matrix} where the $k$-th column corresponds to the score vector for the $k$-th sample, and the predicted labels are elements of an $m$-sized vector $\hat{Y} = \mathbf{Argmax}(f(X))$, where $\mathbf{Argmax}$ returns the indices of the maximum values along columns of a given matrix.
	
	A learning task of classification is to seek a classifier $f_\theta$ in a \textit{function space} $F=\{f_\theta | \theta \in \Theta \}$, where $\theta$ is a set of parameters in a space $\Theta$.
	We refer to $\theta$ as \textit{model parameters} or simply \textit{model}. 
	The task essentially is to find a model $\theta$ that optimizes some objective.

	\subsubsection*{Federated learning}
	FL \cite{mcmahan2016federated} is a framework for collaborative ML consisting of $n$ clients and an FL server.
	Concretely, in each training round $t$, the server selects a subset $I^t$ of clients and broadcasts a \textit{global model} $\theta^t$ among selected clients.
	Each client $i\in I^t$ then trains $\theta^t$ using its own training data to derive a \textit{local model} $\theta^t_i$ and uploads it to the server.
	The server aggregates all the local models $\Theta_{I^t}=\{\theta^t_i \big\bracevert i\in I^t\}$ into a new global model $\theta^{t+1}=\sum_{i\in I^t} \omega^t_i \theta^t_i$, where $\omega^t_i \geq 0$ denotes the aggregation weight assigned to client $i$ for round $t$.
	Finally, after finishing $T$ training rounds, all clients obtain a final model $\theta^{T+1}$.
	
	\subsubsection*{Shapley value}
	The SV \cite{shapley1953sv} is a classic metric for evaluating a player's contribution to a coalition in collaborative game theory.
	It satisfies certain plausible properties in terms of fairness, including balance, symmetry, additivity, and zero element.
	Given clients $I =\{1,...,n\}$, the SV measures the expected marginal utility improvement by each client $i$ over all subsets $S$ of $I$:
	\begin{equation*}
		\label{eq:sv}
		SV_i = \sum\nolimits_{S\subseteq I \backslash\{i\}} \frac{|S|!(n-|S|-1)!}{n!}\big(v(S\cup \{i\}) - v(S)\big),
	\end{equation*}
	where $v(\cdot)$ is some utility function of a set of clients.
	
	\subsubsection*{Neural network}
	We provide an abstraction of \textit{neural network} (NN), the type of classifier considered throughout this paper.
	Consider a batch of samples with a feature matrix $X$.
	An NN $f_\theta$ consists of $L$ layers, where the $l$-th layer is a linear function $lin^{(l)}$ of the model parameters $\theta^{(l)} \subseteq \theta$ and input features $X^{(l)}$ of the layer, where $X^{(1)}=X$.
	For example, convolutional layers and fully-connected layers are typical linear layers.
	When $l<L$, the output features $\hat{Y}^{(l)}=lin^{(l)}(\theta^{(l)}, X^{(l)})$ is processed by an activation function $ac^{(l)}$ (e.g., the \textit{ReLU} and \textit{SoftMax} functions) that takes $\hat{Y}^{(l)}$ as input and outputs $X^{(l+1)}$, which is the input features of the $(l+1)$-th layer.
	Finally, $\hat{Y}^{(L)}$ is the score matrix for the given samples.
	Other classifiers may also fit this abstraction, e.g., logistic classifiers and SVM classifiers, which can be considered one-layer NNs.
	
	\subsubsection*{Homomorphic encryption}
	HE is a cryptographic primitive that enables arithmetic operations over encrypted data.
	Given a plaintext $pt$, we denote its ciphertext under HE as $\llbracket pt \rrbracket$.
	A \textit{fully HE} (FHE) system can support additions and multiplications between ciphertexts or between a ciphertext and a plaintext.
	A modern HE system, such as  \textit{CKKS} \cite{cheon2017ckks}, usually provides the single instruction multiple data (SIMD) functionality:
	a ciphertext of CKKS has $N$ \textit{ciphertext slots} to store scalars, where $N$ is a constant integer decided by the parameters of the HE system;
	it supports homomorphic entrywise addition $\oplus$ and multiplication \changemarker{(HMult)} $\odot$ between two encrypted vectors (or between an encrypted vector and a plaintext vector), which are almost as efficient as additions and multiplications between two encrypted scalars (or between an encrypted scalar and a plaintext scalar), respectively;
	it can also rotate an encrypted vector $\llbracket pt \rrbracket$ $j$ steps by a \textit{homomorphic rotation} \changemarker{(HRot)} $Rot(\llbracket pt \rrbracket, j)$.
	
	As each layer of an NN is a linear function $lin^{(l)}$ of input features $X^{(l)}$ and model parameters $\theta^{(l)}$, we can homomorphically evaluate it by implementing additions $\oplus$ and multiplications $\odot$ between/with $\llbracket X^{(l)} \rrbracket$ and $\llbracket \theta^{(l)} \rrbracket$. 
	Using SIMD, we can pack a batch of test samples as a matrix and simultaneously perform homomorphic operations over them.\footnote{A matrix's ciphertext is derived by horizontally scanning it into a vector.}
	However, exactly evaluating an activation function under HE is often difficult since it is usually nonlinear.
	
	\subsubsection*{Secure matrix multiplication}
	Some types of linear layers involve matrix multiplications, which HE does not directly support.
	For example, a fully-connected layer is a matrix multiplication between a matrix of model parameters and a matrix of input features.
	To homomorphically evaluate a matrix multiplication, we need to transform it into a series of entrywise additions and multiplications that can be directly evaluated under HE.
	
	Throughout this paper, when discussing secure matrix multiplication, we consider evaluating $AB$, where $A$ is a $d_{out} \times d_{in}$ matrix of model parameters, and $B$ is a $d_{in} \times m$ matrix of input features of $m$ samples; 
	$d_{in}$ and $d_{out}$ may be a linear layer's input and output sizes, respectively.
	For ease of discussion, we suppose $d_{out} \leq d_{in}$ without loss of generality.
	
	Let us define some notations for matrix operations.
	Given a $d_1 \times d_2$ matrix $\mathcal{M}$, $\mathcal{M}[j, k]$ denotes the $((k-1) \bmod d_2 + 1)$-th entry of the $((j-1) \bmod d_1 +1)$-th row of $\mathcal{M}$, and $\mathcal{M}[j_1:j_2, k_1:k_2]$ denotes the submatrix of $\mathcal{M}$ derived by extracting its $((k_1-1) \bmod d_2 + 1)$-th to $((k_2-1) \bmod d_2 + 1)$-th columns of its $((j_1-1) \bmod d_1 + 1)$-th to $((j_2-1) \bmod d_1 + 1)$-th rows.
	We use $\mathcal{M}_1;\mathcal{M}_2$ to denote a vertical concatenation of matrices $\mathcal{M}_1$ and $\mathcal{M}_2$ (if they have the same number of columns) and $\mathcal{M}_1 \big\bracevert \mathcal{M}_2$ to denote a horizontal concatenation of them (if they have the same number of rows).
	We also define four linear transformations $\sigma$, $\tau$, $\xi$, and $\psi$ that given a matrix $\mathcal{M}$ yield a transformed matrix of the same shape:
	\begin{align*}
		\forall j, k, \hspace{0.25em} & \sigma(\mathcal{M})[j, k] = \mathcal{M}[j, j+k], \tau(\mathcal{M})[j, k] = \mathcal{M}[j+k, k], \\
		& \xi(\mathcal{M})[j, k] = \mathcal{M}[j, k+1], \psi(\mathcal{M})[j, k] = \mathcal{M}[j+1, k]. 
	\end{align*}
	Additionally, when a superscript number $o$ is assigned to a linear transformation, it means applying the transformation $o$ times.
	
	\subsubsection*{Additive secret sharing}
	ASS \cite{shamir1979ass} protects a \textit{secret} by splitting it into multiple \textit{secret shares} such that they can be used to reconstruct the secret.
	In this paper, we only need to split a secret into two shares.
	Concretely, given a secret $s$ in a finite field $\mathbb{Z}_p$, where $p$ is a prime, we generate a uniformly random mask $r \in \mathbb{Z}_p$.
	Therefore, $s'=r$ and $s''=(s-r) \bmod p$ are the two shares of $s$.
	To reconstruct the secret, we simply need to add up the shares in the field, i.e., $s \equiv (s' + s'') \bmod p$.
	For a real-number secret, we can encode it as an integer, split it into two integer secret shares, and finally decode them to derive two real-number shares.
	The method for implementing secret sharing on real numbers can be found in \cite{riazi2018chameleon}.
	
	\section{Problem Formulation}
	\label{sec:problem}
	
	\subsubsection*{System model}
	We study secure SV calculation in cross-silo FL.
	We consider the scenario where $n$ organizations, i.e., clients $I=\{1,...,n\}$, lack sufficient training data and join FL to train an accurate prediction model for their own use.
	They run $T$ training rounds with the help of a parameter server.
	
	Considering that clients' training data may vary in size, quality, and distribution, the server needs to evaluate their contributions to the accuracy of the final model $\theta^{T+1}$ after training.
	For fair evaluation, each client $i$ contributes a test set of $m_i$ samples $D_i=(X_i, Y_i)$, where $X_i$ is a $d \times m_i$ feature matrix and $Y_i$ is an $m_i$-sized vector of ground-truth labels; 
	the $k$-th column of $X_i$ and $k$-th entry of $Y_i$ is the feature vector and ground-truth label of her $k$-th sample, respectively.
	The server evaluates the contributions based on a collective test set $\mathcal{D}\!=\!(D_1,...,D_n)$, which has $M=\sum_{i=1}^{n} m_i$ samples.
	
	However, calculating the original SV is impractical in FL.
	Intuitively, to calculate the SV, we must enumerate all subsets of clients, perform $T$ rounds of FL for each subset to obtain a final model, and test all the final models.
	Retraining the models is significantly expensive for computation and communication \cite{wang2020principled}, let alone securely training them;
	it may also cause extra privacy leakage owing to the more models to be released.
	Moreover, the SV assumes that the model utility is independent of the participation order of the clients, which does not stand true because clients may participate in or drop out from FL halfway \cite{wang2020principled}.
	
	Hence, we adopt the \textit{Federated SV} (FSV) \cite{wang2020principled}, a variant of the SV that addresses the aforementioned limitations well and guarantees the same advantageous properties as the SV. 
	The FSV is based on the concept that a client's contribution to the final model $\theta^{T+1}$ is the sum of her contribution to each training round. 
	Let $\theta^t_{S}$ denote the model aggregated from the local models of a set $S$ of clients:
	\begin{equation*}
		\theta^t_{S}=\begin{cases}
			\sum_{i\in S} \omega^t_{i|S} \cdot \theta^t_i &\text{if }S\neq \varnothing,\\
			\theta^t &\text{if }S = \varnothing,
		\end{cases}
	\end{equation*}
	where $\omega^t_{i|S}\geq 0$ is an aggregation weight assigned to client $i$ w.r.t. set $S$ and determined by the training algorithm (e.g., FedAvg \cite{mcmahan2016federated}).  
	When $|S|\!>\!1$, we term $\theta^t_{S}$ as an \textit{aggregated model}.
	Then, for each round $t$, the server evaluates each client' \textit{single-round SV} (SSV):
	\begin{equation*}
		\phi^t_i = \begin{cases}
			\sum_{S\subseteq I^t \backslash\{i\}} \frac{|S|!(|I^t|-|S|-1)!}{|I^t|!}\big(v(\theta_{S\cup \{i\}}) \!- \!v(\theta_{S})\big) &\text{if }i\in I^t,\\
			0 &\text{if }i\notin I^t,
		\end{cases}
	\end{equation*}
	where $v(\theta_{S})$ is the prediction accuracy of model $\theta_{S}$.
	Finally, the server aggregates each client $i$' SSVs into her FSV: $\phi_i = \sum\nolimits_{t\in [T]} \phi^t_i$.
	
	\subsubsection*{Threat model}
	Similar to prior works \cite{shokri2015privacy,zhang2020batchcrypt,liu2020secure}, we assume that all the parties are honest-but-curious and noncolluding because they are organizations complying with laws.
	Our problem is to design a protocol where the server coordinates the calculation of FSVs given encrypted local models while no party learns other parties' private information.  
	We may include an auxiliary server to assist the principal server, which is a common model in secure computation literature.
	In this case, we reasonably assume that both servers, e.g., two cloud service providers, are honest-but-curious and will not collude with each other because collusion puts their business reputations at risk.
	Notably, a server is not a machine but a party that may possess multiple machines for parallel computing.
	
	\subsubsection*{Privacy model}
	The private information considered in this study includes the test data and model parameters;
	the model structure decided by the clients is commonly assumed nonsensitive \cite{juvekar2018gazelle}.
	Readers may consider the test data as public information and question the need to protect them.
	Although a public test set is usually available to researchers for evaluating an ML algorithm or model, numerous practical cases exist where the test data are private:
	\begin{itemize}[leftmargin=*]
		\item In collaborative ML marketplaces \cite{ohrimenko2019collaborative, sim2020collaborative, han2020replication}, clients submit their own test data as a specification of the model they want to collaboratively train and purchase.
		\item In federated evaluation \cite{wang2019federated, paulik2021federated}, Apple and Google let users compute some performance metrics of FL models on their own test sets to improve the user experience. 
		\item For personalized cross-silo FL \cite{fallah2020personalized, huang2021personalized}, researchers assume that clients possess non-IID test data.
	\end{itemize} 
	The test data may contain clients' private information and be their proprietary assets, so they intuitively need to be protected.

	\section{One-Server Protocol: HESV}
	In this section, we present a one-server solution to secure SV calculation, named HESV (\underline{HE}-Based \underline{S}hapley \underline{V}alue).
	
	\subsection{Secure testing based purely on HE}
	HESV employs a purely HE-based privacy protection scheme that encrypts both model parameters and test data using HE, as described in Algorithm \ref{alg:hesv}. 
	To begin, each client encrypts her test data and local models using HE \changemarker{(Step 3)}.
	Then, for each training round $t$, the server enumerates all subsets of the selected clients $I^t$ \changemarker{(Step 5)};
	for each subset $S$, he aggregates the corresponding encrypted local models into $\llbracket \theta^t_S \rrbracket$ \changemarker{(Step 6)}, runs Algorithm \ref{alg:infer_hesv} to count the correct predictions made by $\llbracket \theta^t_S \rrbracket$ \changemarker{(Step 7)}, and derives model utility $v(\theta^t_S)$ \changemarker{(Step 8)}.
	Finally, clients' SSVs and FSVs are computed based on the utilities of local and aggregate models \changemarker{(Steps 9 and 10)}.
	
	\begin{figure}
		\small
		\begin{minipage}{\columnwidth}
			\removelatexerror
			\begin{algorithm}[H]
				\caption{One-Server Protocol: HESV}
				\begin{algorithmic}[1]
					\STATE Server: Randomly select a leader client
					\STATE Leader: Generate a public key $pk$ and a private key $sk$ of HE and broadcast them among the other clients
					\STATE Each client $i$: Encrypt her test data and local models and upload them
					\FOR{$t$ in $\{1,...,T\}$}
					\FOR {$S \subseteq I^t$ s.t. $|S| \geq 1$}
					\STATE Server: Compute $\llbracket \theta^t_S \rrbracket$ by aggregation under HE 
					\STATE Server: Run $\Pi_{HE}(\llbracket \theta^t_S \rrbracket, \llbracket X_i \rrbracket,\llbracket Y_i \rrbracket)$ to obtain $cnt_i$ for all $i \in I$
					\STATE Server: Calculate $v(\theta^t_S) = (cnt_1+...+cnt_n)/M$
					\ENDFOR
					\ENDFOR
					\STATE Server: Calculate SSVs $\phi^{t}_1,...,\phi^{t}_n, \forall t \in [T]$
					\STATE Server: Calculate FSVs $\phi_1,...,\phi_n$
				\end{algorithmic}
				\label{alg:hesv}
			\end{algorithm}
		\end{minipage}
		\begin{minipage}{\columnwidth}
			\removelatexerror
			\begin{algorithm}[H]
				\caption{$\Pi_{HE}$: Secure Testing for HESV}
				\begin{algorithmic}[1]
					\REQUIRE encrypted model $\llbracket \theta \rrbracket$, features $\llbracket X \rrbracket$, and labels $\llbracket Y \rrbracket$
					\ENSURE count $cnt$ of correct predictions 
					\FOR {each layer $l\in \{1,...,L\}$ of model $\theta$}
					\STATE Server: Calculate $\llbracket \hat{Y}^{(l)} \rrbracket = lin^{(l)}(\llbracket \theta^{(l)} \rrbracket, \llbracket X^{(l)} \rrbracket)$ and send $\llbracket \hat{Y}^{(l)} \rrbracket$ to client $i_{l+1} \neq i_{l}$, where $\llbracket X^{(1)} \rrbracket = \llbracket X \rrbracket$ 
					\STATE Client $i_{l+1}$: Decrypt $\llbracket \hat{Y}^{(l)} \rrbracket$
					\IF{$l < L$}
					\STATE Client $i_{l+1}$: Compute $X^{(l+1)}= ac^{(l)}(\hat{Y}^{(l)})$, and upload $\llbracket X^{(l+1)} \rrbracket$
					\ELSE 
					\STATE Client $i_{L+1}$: Compute $\hat{Y}=\mathbf{Argmax}(\hat{Y}^{(L)})$, and upload $\llbracket \hat{Y} \rrbracket$
					\ENDIF
					\ENDFOR
					\STATE Server: Calculate $\llbracket \tilde{Y} \rrbracket = -1 \odot \llbracket \hat{Y} \rrbracket \oplus \llbracket Y \rrbracket$ and send it to client $i_{L+2}$
					\STATE Client $i_{L+2}$: Decrypt $\llbracket \tilde{Y} \rrbracket$ and upload $cnt=\sum_{k=1}^{m} \boldsymbol{1}(|\tilde{Y}[k]| < 0.5)$
				\end{algorithmic}
				\label{alg:infer_hesv}
			\end{algorithm}
		\end{minipage}
	\end{figure}
	
	Considering that HE supports nonlinear activations poorly, HESV adopts the \textit{globally-encrypted-locally-decrypted} strategy \cite{juvekar2018gazelle,zheng2018gelu}: linear layers are homomorphically evaluated on the server's side, whereas activation functions are calculated on the clients' side without encryption.
	As depicted in Algorithm \ref{alg:infer_hesv}, for each layer $l$, there is a client $i_l$ holding input features $X^{(l)}$.
	The server then evaluates the linear function $lin^{(l)}$ by applying c2c multiplications/additions between/with the encrypted input $\llbracket X^{(l)} \rrbracket$ \changemarker{(Step 2)} and model parameters $\llbracket \theta^{(l)} \rrbracket$ and sends the output features $\llbracket \hat{Y}^{(l)} \rrbracket$ to client $i_{l+1}$ for decryption \changemarker{(Step 3)}.
	Clients $i_{l}$ and $i_{l+1}$ should be different entities owing to a security issue that will be discussed in Section \ref{sec:security}.
	If $l<L$, client $i_{l+1}$ applies the activation function $ac^{(l)}(\hat{Y}^{(l)})$ to obtain the input $X^{(l+1)}$ of the subsequent layer \changemarker{(Step 5)};
	otherwise, she calculates the predicted labels $\hat{Y}$ \changemarker{(Step 7)}.
	Finally, the server computes the differences $\llbracket \tilde{Y} \rrbracket$ between the predictions $\llbracket \hat{Y} \rrbracket$ and ground-truth labels $\llbracket Y \rrbracket$ 
	and counts correct predictions with the help of some client $i_{L+2}$ \changemarker{(Steps 8 and 9)}.
	Considering that HE introduces slight noise into ciphertexts, to tolerate the noise, we identify correct predictions by judging whether the absolute difference $|\tilde{Y}[k]|$ is smaller than $0.5$ rather than whether $|\tilde{Y}[k]| = 0$.

	\subsection{Matrix Squaring}
	We propose an extension to the SOTA method called \textit{Matrix Squaring} for homomorphic matrix multiplications under HESV.
	
	\subsubsection{SOTA method}
	\label{sec:mat_square_original}
	
	When $d_{in} \leq \lfloor \sqrt{N} \rfloor$, the SOTA method \cite{jiang2018e2dm} supports computing the matrix product $AB$ under HE.
	Suppose that $d_{out}$ exactly divides $d_{in}$.\footnote{If $d_{out}$ divides $d_{in}$ with a remainder, we can pad $A$ with zero-valued rows to obtain a $d'_{out} \times d_{in}$ matrix such that $d'_{out}$ exactly divides $d_{in}$.}
	This method evaluates $AB$ as follows:
	\begin{enumerate}[leftmargin=*]
		\item \textbf{Squaring}: We obtain two square matrices $\bar{A}$ and $\bar{B}$ of order $d_{in}$.
		The $\bar{A}$ matrix vertically packs $d_{in} / d_{out}$ copies of $A$, i.e., $\bar{A} = (A; ...; A)$, while $\bar{B}$ is derived by padding $(d_{in} - m)$ zero-valued columns $\boldsymbol{0}_{d_{in}\times (d_{in} - m)}$ to the end edge of $B$, i.e.,  $\bar{B} = (B \big\bracevert \boldsymbol{0}_{d_{in}\times (d_{in} - m)})$.
		\item \textbf{Linear transformation}: We linearly transform $\bar{A}$ and $\bar{B}$ into two sets of matrices $\{\bar{A}^{(o)}\}_{o=1}^{d_{out}}$ and $\{\bar{B}^{(o)}\}_{o=1}^{d_{out}}$, respectively. 
		Matrices $\bar{A}^{(1)}$ and $\bar{B}^{(1)}$ are derived by rotating the $j$-th row of $\bar{A}$ $j-1$ steps for all $j \in [d_{in}]$ and rotating the $k$-th column of $\bar{B}$ $k-1$ steps for all $k \in [d_{in}]$, respectively, i.e., $\bar{A}^{(1)} = \sigma(\bar{A})$, and $\bar{B}^{(1)} = \tau(\bar{B})$.
		Then, for each $o \in [2, d_{out}]$, we can shift $\bar{A}^{(1)}$ $o-1$ columns and $\bar{B}^{(1)}$ $o-1$ rows to obtain the $\bar{A}^{(o)}$ and $\bar{B}^{(o)}$ matrices, i.e., $\bar{A}^{(o)} = \xi^{(o-1)}(\bar{A}^{(1)})$, and $\bar{B}^{(o)} = \psi^{(o-1)}(\bar{B}^{(1)})$.
		\item \textbf{Encryption}: We encrypt the transformed matrices $\{\bar{A}^{(o)}\}_{o=1}^{d_{out}}$ and $\{\bar{B}^{(o)}\}_{o=1}^{d_{out}}$ and upload them to the server.
		\item \textbf{Entrywise operations}: The server computes $\llbracket H \rrbracket = \llbracket \bar{A}^{(1)} \rrbracket \odot \llbracket \bar{B}^{(1)} \rrbracket \oplus ... \oplus \llbracket \bar{A}^{(d_{out})} \rrbracket \odot \llbracket \bar{B}^{(d_{out})} \rrbracket$.
		\item \textbf{Rotation and extraction}: The matrix product $AB$ can be obtained by vertically splitting matrix $H$ into $d_{in} / d_{out}$ submatrices, adding them up, and extracting the first $m$ columns of the result, i.e., $AB=\tilde{H}[1:d_{out}, 1:m]$,
		where $        \llbracket \tilde{H} \rrbracket = \oplus_{o=0}^{d_{in} / d_{out} - 1} Rot(\llbracket H \rrbracket, d_{out} \cdot d_{in} \cdot o)$.\footnote{For computing $\llbracket \tilde{H} \rrbracket$, the server can apply a repeated doubling approach to improve efficiency \cite{jiang2018e2dm}.}
	\end{enumerate}
	
	However, a ciphertext of HE does not have sufficient slots to store a large matrix of order $d_{in}> \lfloor \sqrt{N} \rfloor$.
	This is a typical case because an NN's input is usually large.
	Even if we can use multiple ciphertexts to store the matrix, slot rotations across ciphertexts are not supported, which makes the SOTA method fail. 
	
	\begin{figure}[t]
		\centering
		\includegraphics[scale=0.6]{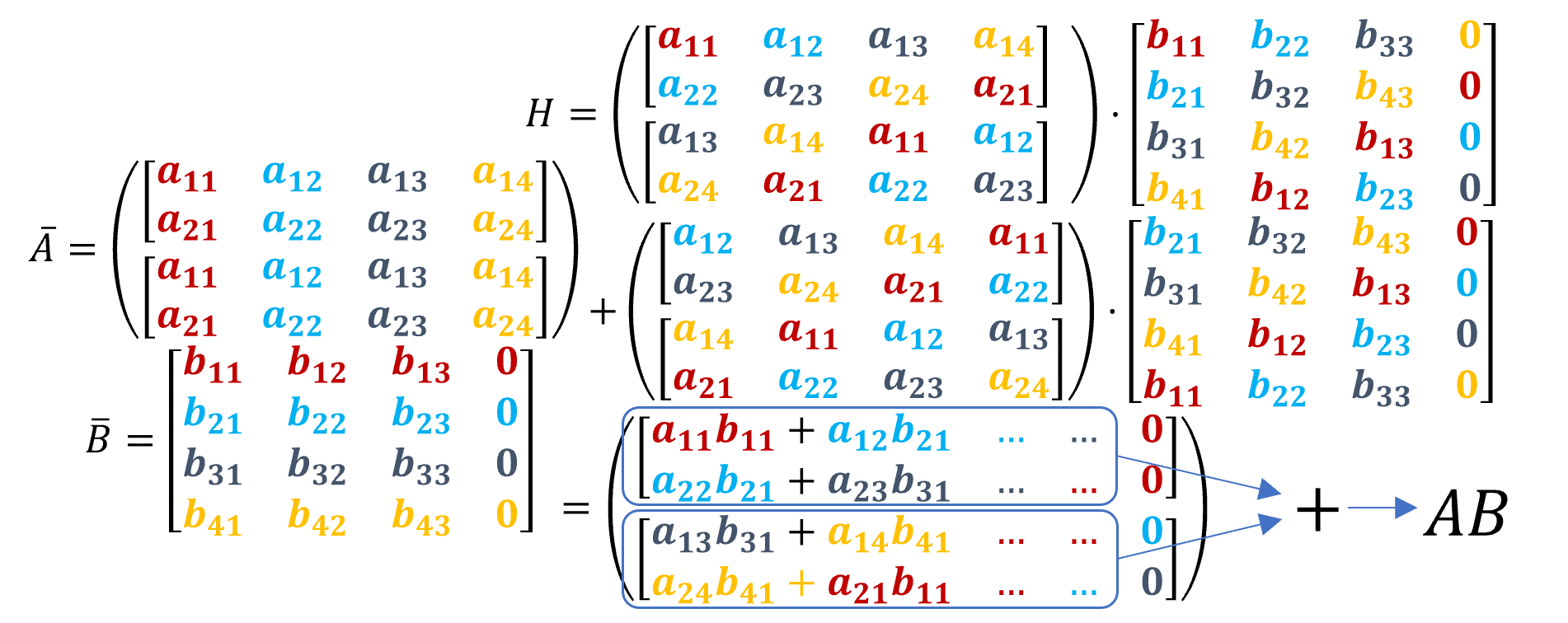}
		\caption{\changemarker{ SOTA method for matrix multiplication.}}
		\label{fig:sota}
	\end{figure}
	
	\begin{example}
		Figure \ref{fig:sota} shows how the SOTA method works for a $2 \times 4$ matrix $A$ and a $4 \times 3$ matrix $B$ with $N = 16$, where $A[i,j]=a_{ij}$, and $B[i,j]=b_{ij}$. 
		First, we vertically pack two copies of $A$ and pad $B$ with zeros to derive $4\times 4$ matrices $\bar{A}$ and $\bar{B}$, respectively.
		We then apply entrywise operations over two pairs of transformed matrices to obtain $H$.
		Essentially, $AB$ is derived by adding $H[1:2,1:3]$ and $H[3:4,1:3]$.
		However, when $N = 12$, the transformed matrices are overly large to encrypt into a ciphertext.
	\end{example}
	
	\subsubsection{Our improvement}
	To address this issue, Matrix Squaring involves dividing matrices $A$ and $B$ into smaller submatrices that can be stored in a ciphertext.
	Suppose that $\lfloor \sqrt{N} \rfloor$ exactly divides $d_{in}$ without loss of generality.\footnote{We can pad $A$ and $B$ with zeros to ensure this condition.}
	Concretely, when $d_{in}> \lfloor \sqrt{N} \rfloor$, we can vertically split $A$ every $\lfloor \sqrt{N} \rfloor$-th column to obtain $K$ submatrices $A_{(\cdot, 1)}, ...,A_{(\cdot, K)}$ and horizontally split $B$ every $\lfloor \sqrt{N} \rfloor$-th row to derive $K$ submatrices $B_{(1, \cdot)}, ...,B_{(K, \cdot)}$, where $K=d_{in} / \lfloor \sqrt{N} \rfloor$, $(A_{(\cdot, 1)}\big\bracevert ... \big\bracevert A_{(\cdot, K)})=A$, and $(B_{(1, \cdot)}; ...;B_{(K, \cdot)})=B$.
	Then, we have
	\begin{equation*}
		AB = \sum_{k=1}^{K} A_{(\cdot, k)}B_{(k, \cdot)}.
	\end{equation*}
	Therefore, we can evaluate $AB$ under HE by applying the SOTA method over $d_{in} / \lfloor \sqrt{N} \rfloor$ pairs of submatrices $A_{(\cdot, k)}$ and $B_{(k, \cdot)}$ and aggregating the results.
	This inherently requires that $m$ should not exceed $\lfloor \sqrt{N} \rfloor$; 
	otherwise, any square matrix transformed from a $\lfloor \sqrt{N} \rfloor \times m$ matrix $B_{(k, \cdot)}$ cannot be encrypted into a single ciphertext.
	Furthermore, when $d_{out}> \lfloor \sqrt{N} \rfloor$, we can horizontally split $A_{(\cdot, k)}$ every $\lfloor \sqrt{N} \rfloor$-th row into $A_{(1,k)},...,A_{(J,k)}$, where $J=\lceil d_{out} / \lfloor \sqrt{N} \rfloor \rceil$, and $(A_{(1,k)};...;A_{(J,k)})=A_{(\cdot, k)}$.
	Then, we have 
	\begin{equation*}
		A_{(\cdot, k)}B_{(k, \cdot)} = (A_{(1,k)}B_{(k, \cdot)};...;A_{(J,k)}B_{(k, \cdot)}).
	\end{equation*}
	Hence, we can evaluate $A_{(\cdot, k)}B_{(k, \cdot)}$ by applying the SOTA method over $J$ pairs of $A_{(j,k)}$ and $B_{(k, \cdot)}$ and vertically packing the results.
	
	\begin{figure}[ht]
		\centering
		\includegraphics[scale=0.4]{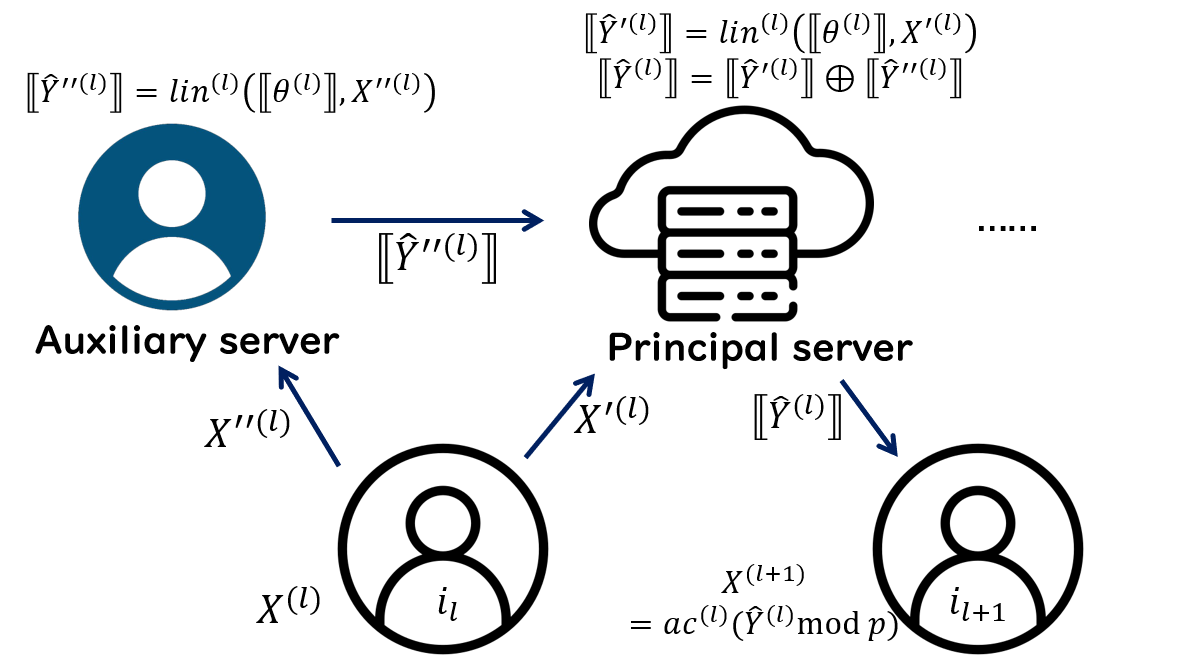}
		\caption{Secure testing for SecSV.}
		\label{fig:secsv}
	\end{figure}
	
	\begin{figure}[ht]
		\small
		\begin{minipage}{\columnwidth}
			\removelatexerror
			\begin{algorithm}[H]
				\caption{SecSV}
				\begin{algorithmic}[1]
					\STATE Server $\mathcal{P}$: Randomly select a leader client
					\STATE Leader: Generate a public key $pk$ and a private key $sk$ of HE and broadcast them among the other clients
					\STATE Each client $i$: Encrypt her own local models and send them to the two servers; then, generate secret shares $D_i', D_i''$ of $D_i$, send $D_i'$ to server $\mathcal{P}$, and send $D_i''$ to server $\mathcal{A}$
					\FOR{$t$ in $\{1,...,T\}$}
					\IF{$\text{skipSamples}==\text{True}$}
					\STATE Server $\mathcal{P}$: Run SampleSkip($\{\llbracket \theta^t_i \rrbracket\}_{\forall i \in I^t}, (\mathcal{D}', \mathcal{D}''))$ to obtain utilities $\{v(\theta^t_S)\}_{\forall S \subseteq I^t, |S| > 0}$
					\ELSE 
					\FOR {$S \subseteq I^t, |S| > 0$}
					\STATE Servers $\mathcal{P}, \mathcal{A}$: Compute $\llbracket \theta^t_S \rrbracket$ by aggregation under HE 
					\STATE Server $\mathcal{P}$: Run $\Pi_{Sec}(\llbracket \theta^t_S \rrbracket, (\mathcal{D}', \mathcal{D}''))$ to obtain $\Phi^t_{S}$
					\STATE Server $\mathcal{P}$: Calculate $v(\theta^t_S) = |\Phi^t_{S}|/M$
					\ENDFOR
					\ENDIF
					\ENDFOR
					\STATE Server $\mathcal{P}$: Calculate SSVs $\phi^{t}_1,...,\phi^{t}_n, \forall t \in [T]$
					\STATE Server $\mathcal{P}$: Calculate FSVs $\phi_1,...,\phi_n$
				\end{algorithmic}
				\label{alg:secsv}
			\end{algorithm}
		\end{minipage}
		\begin{minipage}{\columnwidth}
			\removelatexerror
			\begin{algorithm}[H]
				\caption{$\Pi_{Sec}$: Secure Testing for SecSV}
				\begin{algorithmic}[1]
					\REQUIRE encrypted model $\llbracket \theta \rrbracket$, and secret shares $\mathcal{D}', \mathcal{D}''$ of $\mathcal{D}$
					\ENSURE IDs $\Phi$ of correctly predicted samples
					\STATE $\Phi\gets \emptyset$
					\FOR {each $D'=(X',Y'), D''=(X'',Y'') \in (\mathcal{D}', \mathcal{D}'')$}
					\FOR {each model layer $l\in \{1,...,L\}$}
					\STATE Server $\mathcal{A}$: Calculate $\llbracket \hat{Y}''^{(l)} \rrbracket = lin^{(l)}(\llbracket \theta^{(l)} \rrbracket, X''^{(l)})$ and send $\llbracket \hat{Y}''^{(l)} \rrbracket$ to server $\mathcal{P}$, where $X''^{(1)} = X''$
					\STATE Server $\mathcal{P}$: Calculate $\llbracket \hat{Y}'^{(l)} \rrbracket = lin^{(l)}(\llbracket \theta^{(l)} \rrbracket, X'^{(l)})$, where $X'^{(1)} = X'$
					\STATE Server $\mathcal{P}$: Send $\!\llbracket \hat{Y}^{(l)} \!\rrbracket \!=\! \llbracket\! \hat{Y}'^{(l)} \!\rrbracket \!\oplus\! \llbracket \!\hat{Y}''^{(l)} \!\rrbracket$ to client $i_{l+1} \neq i_{l}$ 
					\STATE Client $i_{l+1}$: Compute $\hat{Y}^{(l)} = Decrypt(\llbracket \hat{Y}^{(l)} \rrbracket) \bmod p$
					\IF{$l < L$}
					\STATE Client $i_{l+1}$: Compute $X^{(l+1)}= ac^{(l)}(\hat{Y}^{(l)})$, generate shares $X'^{(l+1)}, X''^{(l+1)}$, and distribute them to servers $\mathcal{P}$ and $\mathcal{A}$
					\ELSE 
					\STATE Client $i_{L+1}$: Calculate $\hat{Y}=\mathbf{Argmax}(\hat{Y}^{(L)})$, generate shares $\hat{Y}', \hat{Y}''$, and distribute them to servers $\mathcal{P}$ and $\mathcal{A}$
					\ENDIF
					\ENDFOR
					\STATE Server $\mathcal{A}$: Compute and send $\tilde{Y}''= \hat{Y}'' - Y''$ to server $\mathcal{P}$
					\STATE Server $\mathcal{P}$: Calculate $\tilde{Y} = abs((\hat{Y}' - Y' + \tilde{Y}'') \bmod p)$, where $abs$ denotes taking the entrywise absolute value
					\STATE Server $\mathcal{P}$: Update the IDs of correctly predicted samples $\Phi \gets \Phi \cup \{id(D[k])| \tilde{Y}[k] < 0.5\}$, where $id(D[k])$ is the ID of  test sample $D[k]$
					\ENDFOR
					\RETURN $\Phi$
				\end{algorithmic}
				\label{alg:infer_secsv}
			\end{algorithm}
		\end{minipage}
	\end{figure}

	\section{Two-Server Protocol: SecSV}
	
	HESV has three considerable drawbacks.
	First, it involves numerous c2c multiplications, which are highly inefficient in computation.
	Second, it cannot fully utilize the SIMD feature of HE.
	Since clients encrypt test samples on their local sides, the server cannot pack samples from different sources, which may waste some ciphertext slots.
	Third, it fully evaluates the entire test set for all aggregated models, which is time-consuming. 
	In this section, we propose a two-server protocol with an auxiliary server named SecSV (\underline{Sec}ure \underline{S}hapley \underline{V}alue) to overcome the drawbacks of HESV. The features of this protocol are (1) a hybrid secure testing scheme, (2) an efficient homomorphic matrix multiplication method, and (3) an acceleration technique for SV calculation.

	\subsection{Hybrid Secure Testing Scheme}
	\label{sec:hybrid_testing}
	SecSV adopts a hybrid scheme for secure testing: it encrypts models by HE but protects test data by ASS.
	An auxiliary server $\mathcal{A}$ is needed to help the principal server $\mathcal{P}$ test models on secretly shared test data.
	Concretely, as shown in Algorithm \ref{alg:secsv}, each client $i$ encrypts her local models by HE but protects her test data $D_i$ by splitting it into two secret shares $D'_i$ and $D''_i$ (Step 3).
	Thereafter, the two servers evaluate the shares $\mathcal{D}', \mathcal{D}''$ of the collective test set $\mathcal{D}=(D_1,...,D_n)$ by running Algorithm \ref{alg:infer_secsv} (Step 10).
	Figure \ref{fig:secsv} shows that because the shares $X'^{(l)}, X''^{(l)}$ of input features $X^{(l)}$ are in plaintext form for all layers $l$, c2c multiplications are avoided.

	Algorithm \ref{alg:infer_secsv} shows how SecSV evaluates an encrypted model.
	For each model layer $l$, server $\mathcal{P}$ holds a share $X'^{(l)}$ of the input features $X^{(l)}$ while server $\mathcal{A}$ possesses the other share $X''^{(l)}$.
	They each evaluate the linear function $lin^{(l)}$ over their own shares to compute shares $\llbracket \hat{Y}'^{(l)} \rrbracket$ and $\llbracket \hat{Y}''^{(l)} \rrbracket$, respectively (Steps 4 and 5).
	Then, after receiving $\llbracket \hat{Y}''^{(l)} \rrbracket$ from server $\mathcal{A}$, server $\mathcal{P}$ adds up $\llbracket \hat{Y}'^{(l)} \rrbracket$ and $\llbracket \hat{Y}''^{(l)} \rrbracket$ to reconstruct the output features $\llbracket \hat{Y}^{(l)} \rrbracket$ (Step 6), which is sent to client $i_{l+1}$ for decryption and modulo (Step 7).
	If $l<L$, client $i_{l+1}$ activates the output features and generates shares $X'^{(l+1)}, X''^{(l+1)}$ of the activated features $X^{(l+1)}$ for evaluating the next layer (Step 9);
	otherwise, client $i_{L+1}$ computes the predicted labels $\hat{Y}$ and generates shares $\hat{Y}', \hat{Y}''$ for comparison with the shares $Y', Y''$ of the ground-truth labels $Y$ (Step 11).
	After obtaining the absolute differences $\tilde{Y}$ between $\hat{Y}$ and $Y$ (Steps 12 and 13), server $\mathcal{P}$ updates an ID set $\Phi$ that contains the IDs of the correctly predicted samples with a tolerable difference $\tilde{Y}[k] < 0.5$ (Step 14).

	\begin{example}
		Consider $4$ clients and $3$-layer models.
		When testing client $1$' local model $\theta^t_1$, given shares $X'^{(1)}, X''^{(1)}$ of input features from client $2$, the servers evaluate layer $1$ under HE, aggregate shares $\hat{Y}'^{(1)}, \hat{Y}''^{(1)}$ of output features, and return $\hat{Y}^{(1)}$ to client $3$ for activation.
		Similarly, given $X'^{(2)}, X''^{(2)}$ from client $3$, the servers evaluate layer $2$ and send $\hat{Y}^{(2)}$ to client $4$.
		Finally, client $3$ obtains the output features $\hat{Y}^{(3)}$, computes the predicted labels $\hat{Y}$, and returns shares $\hat{Y}', \hat{Y}''$ to the servers for comparison with the shares $Y', Y''$ of ground-truth labels.
	\end{example}

	\subsection{Matrix Reducing}
	
	\subsubsection*{Our insight}
	We then propose a secure matrix multiplication method named \textit{Matrix Reducing}.
	To evaluate $AB$ under HE, we have to transform $A$ and $B$ such that they have the same shape because HE only allows entrywise operations.
	For example, the SOTA method squares and linearly transforms $A$ and $B$ to derive pairs of square matrices, applies entrywise multiplications between each pair of matrices, and aggregates the multiplication results.
	However, the aggregated square matrix $H$ (see Section \ref{sec:mat_square_original}) is not the matrix product $AB$, and the SOTA method further needs to rotate $H$ to aggregate its submatrices, which is expensive under HE.
	This inefficiency results from the need to square $A$ and $B$:
	the SOTA method squares a rectangular $A$ by packing multiple copies of $A$ together, which enables parallel computation over $A$ and avoids the waste of some ciphertext slots.
	
	Considering that the matrix product $AB$ is a $d_{out} \times m$ matrix, we introduce the concept of transforming $A$ and $B$ into $d_{out} \times m$ matrices to avoid homomorphic rotations.
	Concretely, in the unencrypted environment, to obtain each entry in the matrix product $AB$, we need $d_{in}$ multiplications between entries in $A$ and $B$; 
	this suggests that we can somehow preprocess $A$ and $B$ to generate $d_{in}$ pairs of $d_{out} \times m$ matrices, apply entrywise multiplications between each pair of matrices, and aggregate the multiplication results to derive $AB$.
	We design Matrix Reducing based on this idea.
	
	\begin{algorithm}[t]
		\small
		\caption{Matrix Reducing}
		\begin{algorithmic}[1]
			\REQUIRE $d_{out} \times d_{in}$ matrix $A$, and $d_{in} \times m$ matrix $B$ ($m \leq \lfloor N/d_{out} \rfloor$)
			\ENSURE encrypted matrix product $\llbracket AB \rrbracket$
			\STATE Client: Horizontally pack $\lceil m/d_{in} \rceil$ copies of $A$ into $\bar{A}=(A\big\bracevert...\big\bracevert A)$
			\STATE Client: Run the first two steps of $\delta(\bar{A}, B)$ to obtain $\{\tilde{A}^{(o)}, \tilde{B}^{(o)}\}_{o=1}^{d_{in}}$
			\STATE Client: Encrypt $\{\tilde{A}^{(o)}\}_{o=1}^{d_{in}}$ and $\{\tilde{B}^{(o)}\}_{o=1}^{d_{in}}$ and upload them
			\STATE Server: Return $\llbracket AB \rrbracket\! = \!\llbracket \tilde{A}^{(1)} \rrbracket \!\odot \!\llbracket \tilde{B}^{(1)} \rrbracket \!\oplus\! ...\! \oplus\! \llbracket \tilde{A}^{(d_{in})} \rrbracket \!\odot \!\llbracket \tilde{B}^{(d_{in})} \rrbracket$
		\end{algorithmic}
		\label{alg:mat_reduce}
	\end{algorithm}
	
	\subsubsection*{Design details}
	To assist in understanding Matrix Reducing, we introduce a function $R = \delta(\bar{A}, \bar{B})$ that takes as input a $d_{out} \times d$ matrix $\bar{A}$ and a $d_{in} \times m$ matrix $\bar{B}$ and outputs a $d_{out} \times m$ matrix $R$, where $d \geq m$ and $d_{in}$ exactly divides $d$.
	$\delta(\bar{A}, \bar{B})$ runs as follows:
	\begin{enumerate}[leftmargin=*]
		\item \textbf{Linear transformation:} We linearly transform $\bar{A}$ and $\bar{B}$ into two sets of matrices $\{\bar{A}^{(o)}\}_{o=1}^{d_{in}}$ and $\{\bar{B}^{(o)}\}_{o=1}^{d_{in}}$, where $\bar{A}^{(o)}=\xi^{(o-1)}(\sigma(\bar{A}))$, and $\bar{B}^{(o)}=\psi^{(o-1)}(\tau(\bar{B}))$.
		\item \textbf{Reduction:} For each $o\in [d_{in}]$, we extract two $d_{out} \times m$ matrices $\tilde{A}^{(o)}$ and $\tilde{B}^{(o)}$ from $\bar{A}^{(o)}$ and $\bar{B}^{(o)}$, respectively, where $\tilde{A}^{(o)} = \bar{A}^{(o)}[1:d_{out}, 1:m]$, and $\tilde{B}^{(o)} = \bar{B}^{(o)}[1:d_{out}, 1:m]$.
		\item \textbf{Element-wise operations:} We apply entrywise multiplications and additions to compute $R=\sum_{o=1}^{d_{in}}\tilde{A}^{(o)} \cdot \tilde{B}^{(o)}$.
	\end{enumerate}
	When $m \leq d_{in}$, according to lemma \ref{lemma:correctness}, we have $AB = \delta(A, B)$.
	Therefore, we can encrypt matrices $\{\tilde{A}^{(o)}\}_{o=1}^{d_{in}}$ and $\{\tilde{B}^{(o)}\}_{o=1}^{d_{in}}$ and apply homomorphic entrywise multiplications and additions over them to compute $\llbracket AB \rrbracket$.
	
	\begin{lemma}[Correctness]
		\label{lemma:correctness}
		We have $AB=\delta(A, B)$ for any $d_{out} \! \times \! d_{in}$ matrix $A$ and $d_{in} \! \times\! m$ matrix $B$, where $m\! \leq\! d_{in}$, and $d_{out} \!\leq \! d_{in}$.\footnote{See \venueforappendix{\ref{appendix:proof}} to find all the missing proofs.}
	\end{lemma}
	
	Then, when $m \in (d_{in}, \lfloor N/d_{out} \rfloor]$, Matrix Reducing evaluates $AB$ by horizontally packing multiple copies of $A$ into a $d_{out} \times m$ matrix.
	Suppose that $d_{in}$ exactly divides $m$ when $m > d_{in}$ for ease of discussion.
	The matrix product $AB$ can be vertically split into $d_{out} \times d_{in}$ matrices $\{AB_{(\cdot, o)}\}_{o=1}^{m / d_{in}}$, i.e., $AB = \big(AB_{(\cdot, 1)} \big\bracevert ...\big\bracevert AB_{(\cdot, m/d_{in})}\big)$, where $B_{(\cdot, o)}\!=\!B\big[1\!:\!d_{in},(o-1)\cdot d_{in}+1 \!:\! o\cdot d_{in}\big], \forall o \!\in\! [m/d_{in}]$.
	According to lemma \ref{lemma:correctness}, we have $AB = \big(\delta(A,B_{(\cdot, 1)}) \big\bracevert ...\big\bracevert \delta(A,B_{(\cdot, m / d_{in})})\big)$.
	Then, according to Lemma \ref{lemma:commutative}, we can conclude that $AB = \delta(\bar{A}, B)$,
	where $\bar{A}=(A \big\bracevert ... \big\bracevert A)$ is a $d_{out} \times m$ matrix that contains $m / d_{in}$ copies of $A$.
	Hence, as shown in Algorithm \ref{alg:mat_reduce}, Matrix Reducing evaluates $\llbracket AB \rrbracket$ by homomorphically computing $\delta(\bar{A}, B)$, which involves only $d_{in}$ homomorphic multiplications and no homomorphic rotation.
	
	\begin{lemma}[Composition]
		\label{lemma:commutative}
		Given any $d_{out} \times d_{in}$ matrix $A$, $d_{in} \times d_{in}$ matrix $B_1$, and $d_{in} \times m'$ matrix $B_2$, where $d_{in} \geq m'$, we have
		\begin{equation*}
			(\delta(A, B_1) \big\bracevert \delta(A, B_2)) = \delta((A \big\bracevert A), (B_1 \big\bracevert B_2))
		\end{equation*}
	\end{lemma}
	
	\begin{figure}[t]
		\centering
		\includegraphics[scale=0.6]{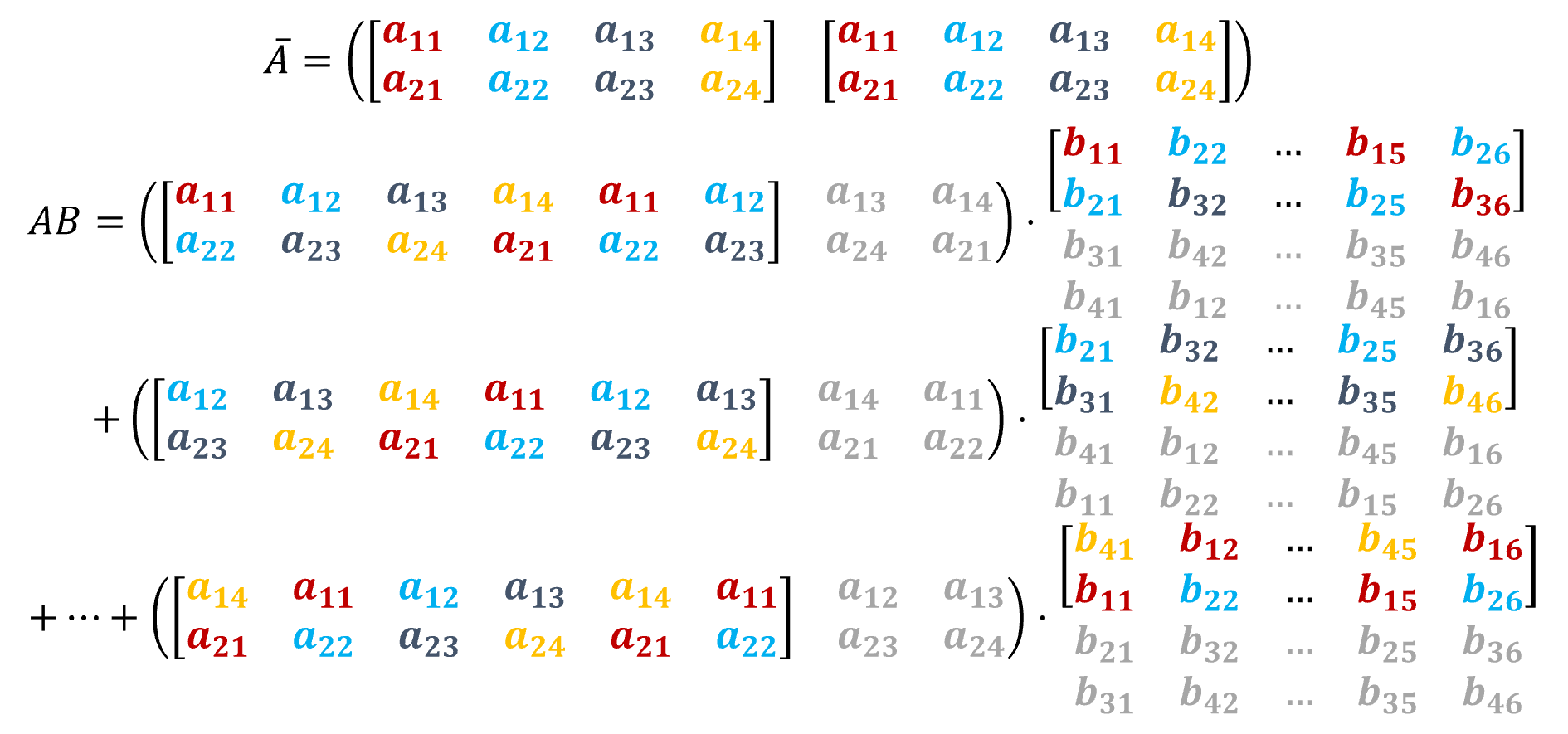}
		\caption{\changemarker{Matrix Reducing for matrix multiplication.}}
		\label{fig:mat_reduce}
	\end{figure}
	
	\begin{example}
		Figure \ref{fig:mat_reduce} depicts how Matrix Reducing works for a $2 \times 4$ matrix $A$ and a $4 \times 6$ matrix $B$, where $A[i,j]=a_{ij}$, and $B[i,j]=b_{ij}$. 
		First, we horizontally pack two copies of $A$ to derive a $2\times 8$ matrix $\bar{A}$.
		Then, we linearly transform $\bar{A}$ and $B$ and reduce the transformed matrices into $2 \times 6$ matrices.
		Finally, we apply entrywise operations over the reduced matrices to derive $AB$.
	\end{example}
	
	\subsection{SampleSkip}
	Finally, we propose to further accelerate secure SV calculation by reducing the number of test samples for secure model testing.
	Concretely, if some test samples are highly likely to be correctly predicted by a model, which we refer to as \textit{skippable samples}, we can skip these samples during model testing to save time.
	We  then identify skippable samples:
	\textit{If a sample can be correctly predicted by two models, it can also be correctly predicted by their aggregation, and thus, it is skippable;
		otherwise, it is a discriminative sample that distinguishes the models' utilities and the clients' contributions.} 
	
	\begin{figure}
		\small
		\begin{minipage}{\columnwidth}
			\removelatexerror
			\begin{algorithm}[H]
				\caption{SampleSkip}
				\begin{algorithmic}[1]
					\REQUIRE encrypted models $\{\llbracket \theta^t_i \rrbracket\}_{\forall i \in I^t}$, secret shares $\mathcal{D}', \mathcal{D}''$ of $\mathcal{D}$
					\ENSURE model utilities $\{v(\theta^t_S)\}_{\forall S \subseteq I^t, |S|>0} $
					\STATE Server $\mathcal{P}$: Initialize a dictionary $\boldsymbol{\Phi}^t$.
					\FOR {$i \in I^t$}
					\STATE Server $\mathcal{P}$: Run $\Pi_{Sec}(\llbracket \theta^t_i \rrbracket, (\mathcal{D}', \mathcal{D}''))$ to obtain $\Phi^t_{\{i\}}$ 
					\STATE Server $\mathcal{P}$: Calculate $v(\theta^t_i) = |\Phi^t_{\{i\}}|/M$
					\ENDFOR
					\FOR {$j=2$ to $j=|I^t|$}
					\FOR {$S \subseteq I^t, |S| = j$}
					\STATE Servers $\mathcal{P}, \mathcal{A}$: Calculate $\llbracket \theta^t_S \rrbracket$ under HE
					\STATE Server $\mathcal{P}$: Compute IDs $\Psi^t_{S}=\text{FindSkippable}(S, \boldsymbol{\Phi}^t)$.
					\STATE Servers $\mathcal{P}, \mathcal{A}$: Drop skippable samples $\Psi^t_{S}$ from $(\mathcal{D}', \mathcal{D}'')$ to obtain discriminative samples $(\mathcal{D}'_{-\Psi^t_{S}}, \mathcal{D}''_{-\Psi^t_{S}})$ 
					\STATE Server $\mathcal{P}$: Run $\Pi_{Sec}(\llbracket \theta^t_S \rrbracket, (\mathcal{D}'_{-\Psi^t_{S}}, \mathcal{D}''_{-\Psi^t_{S}}))$ to obtain $\Phi^t_{S}$ 
					\STATE Server $\mathcal{P}$: $\Phi^t_{S} \gets \Phi^t_{S} \cup \Psi^t_{S}$
					\STATE Server $\mathcal{P}$: Calculate $v(\theta^t_S) = |\Phi^t_{S}|/M$
					\ENDFOR
					\ENDFOR
					\RETURN $\{v(\theta^t_S)\}_{\forall S \subseteq I^t, |S|>0} $
				\end{algorithmic}
				\label{alg:sampleskip}
			\end{algorithm}
		\end{minipage}
		\begin{minipage}{\columnwidth}
			\removelatexerror
			\begin{algorithm}[H]
				\caption{FindSkippable}
				\begin{algorithmic}[1]
					\REQUIRE a set $S$ of clients, a dictionary $\Phi^t$ 
					\ENSURE IDs $\Psi^t_S$ of skippable samples for model $\theta^t_S$
					\STATE $\Psi^t_S \gets \emptyset$
					\FOR{each pair of nonempty subsets $(S', S \setminus S')$ of $S$}
					\IF{$\Phi^t_{S'}, \Phi^t_{S \setminus S'}$ exist}
					\STATE $\Psi^t_S \gets \Psi^t_S \cup (\Phi^t_{S'}\cap \Phi^t_{S \setminus S'})$
					\ENDIF
					\ENDFOR
					\RETURN $\Psi^t_S$
				\end{algorithmic}
				\label{alg:find_skip}
			\end{algorithm}
		\end{minipage}
	\end{figure}
	
	Hence, we propose \textit{SampleSkip} to speed up secure SV calculation by skipping test samples.
	Let $\Phi^t_{S}$ denote an ID set that contains the IDs of the test samples that can be correctly predicted by model $\theta^t_{S}$, where $S\subseteq I^t$ is a subset of clients.
	As shown in Algorithm \ref{alg:sampleskip}, in each round $t$, server $\mathcal{P}$ first initializes a dictionary $\boldsymbol{\Phi}^t$ to record $\Phi^t_{S}$ for all subsets $S\subseteq I^t$ (Step 1).
	For each local model $\theta^t_i, i\in I^t$, the servers test it on the entire test set $(\mathcal{D}', \mathcal{D}'')$ and record $\Phi^t_{\{i\}}$ (Step 3).
	Then, for each aggregated model $\theta^t_{S}, |S|\geq 2$, the servers drop the skippable samples $\Psi^t_{S}$ from $(\mathcal{D}', \mathcal{D}'')$ and test the model on the discriminative samples
	$(\mathcal{D}'_{-\Psi^t_{S}}, \mathcal{D}''_{-\Psi^t_{S}})$ (Steps 8-10);
	the skippable samples $\Psi^t_{S}$ are counted as correctly predicted samples (Step 11).
	Note that SampleSkip also applies to nonsecure SV calculation.
	
	We also propose Algorithm \ref{alg:find_skip} that identifies skippable samples for aggregated models.
	Given that an aggregated model $\theta^t_{S}$ can be aggregated from multiple pairs of models, we need to find the union of the skippable samples determined by each pair of models.
	Hence, Algorithm \ref{alg:find_skip} enumerates all possible pairs of nonempty subsets $(S', S\setminus S')$ of the given set $S$, identifies skippable samples for each pair, and returns the union set of skippable samples;
	this suggests that SampleSkip should iterate over all the subsets $S\subseteq I^t, |S|>1$ in ascending order of their sizes such that $\Phi^t_{S'}$ and $\Phi^t_{S \setminus S'}$ are already recorded in $\Phi^t$ for identifying $\Psi^t_{S}$.
	Notably, although this algorithm takes $O(2^{|S|})$ complexity, it hardly burdens SecSV because the time for running it is far less than that for evaluating models under HE.
	
	\begin{example}
		Consider $3$ clients and $4$ test samples $x_1,...,x_4$. 
		If models $\theta^t_{\{1\}}, \theta^t_{\{2, 3\}}$ correctly predict $x_1, x_2$, models $\theta^t_{\{2\}}, \theta^t_{\{1, 3\}}$ correctly predict $x_2, x_3$, and models $\theta^t_{\{3\}}, \theta^t_{\{1, 2\}}$ correctly predict $x_3, x_4$, then the servers can skip all the test samples when testing $\theta^t_{\{1, 2, 3\}}$.
	\end{example}
	
	\section{Theoretical Analysis and Discussion}
	\subsection{Theoretical Analysis}
	\subsubsection{Time cost}
	Table \ref{tab:matmul} summarizes the time complexities of Matrix Squaring and Matrix Reducing.
	Note that the batch size $m$ is a hyperparameter set by the server(s);
	the requirements on $m$ suggest the maximum batch size such that a batch of samples can be stored in a ciphertext, which maximizes computation efficiency.
	Considering that the maximum batch size may differ between the two methods, the complexities should be averaged over the maximum batch sizes for a fair comparison.
	Evidently, Matrix Reducing outperforms Matrix Squaring.
	First, unlike Matrix Squaring, Matrix Reducing does not require any homomorphic rotation.
	Regarding homomorphic multiplications, Matrix Squaring and Matrix Reducing take $O(d_{in} \cdot d_{out} / \sqrt{N} )$ and $O(d_{in})$ complexity for evaluating $m \leq \min\{d_{in},\lfloor \sqrt{N} \rfloor\}$ and $m \leq \lfloor N/d_{out} \rfloor$ samples, respectively;
	thus, when we batch $\min\{d_{in},\lfloor \sqrt{N} \rfloor\}$ samples for the former method and $\lfloor N/d_{out} \rfloor$ samples for the latter, the latter's complexity $O(d_{in}\cdot d_{out}/N)$ for each sample is better than the former's complexity $O(d_{in} \cdot d_{out} / \sqrt{N} / \min\{d_{in},\lfloor \sqrt{N} \rfloor\})$.
	The uses of the different matrix multiplication methods result in the different time costs of HESV and SecSV, as characterized below.
	
	\begin{table}[t]
		\small
		\centering
		\caption{Comparison of time cost for evaluating $AB$.}
		\begin{tabular}{|l|l|l|}
			\hline
			& Matrix Squaring & Matrix Reducing \\ \hline
			Batch size $m$ & $m \leq \min\{d_{in},\lfloor \sqrt{N} \rfloor\}$ & $m \leq \lfloor N/d_{out} \rfloor$ \\ \hline
			Complexity of HMult & $O(d_{in} \cdot d_{out} / \sqrt{N})$ & $O(d_{in})$ \\ \hline
			Complexity of HRot & $O(d_{in}/(d_{out} \bmod \sqrt{N}))$ & 0 \\ \hline
		\end{tabular}
		\label{tab:matmul}
	\end{table}
	
	\begin{lemma}
		\label{lemma:time_cost}
		Consider models with $L$ layers of matrix multiplication where the input and output sizes of layer $l$ are $d^{(l)}_{in}$ and $d^{(l)}_{out}$, respectively.
		For secure SV calculation, HESV needs $O( \frac{2^n\cdot M \cdot T \cdot (\sum_{l=1}^{L} d^{(l)}_{in} \cdot d^{(l)}_{out})}{\min\{d^{(1)}_{in},...,d^{(L)}_{in},\sqrt{N}\} \cdot \sqrt{N}})$ HMults and $O( \frac{2^n\cdot M \cdot T \cdot \sum_{l=1}^{L} \frac{d_{in}}{d_{out} \bmod \sqrt{N}}}{\min\{d^{(1)}_{in},...,d^{(L)}_{in},\sqrt{N}\}} )$ HRots, while SecSV needs $O( \frac{2^n\cdot M \cdot T \cdot \max\{d^{(1)}_{out},...,d^{(L)}_{out}\} \sum_{l=1}^{L} d^{(l)}_{in}}{N})$ HMults and no HRot.
	\end{lemma}
	
	\subsubsection{Error of SampleSkip}
	According to Theorem \ref{thm:sampleskip}, for \textit{linear classifiers}, e.g., the widely-used logistic classifiers and linear SVM classifiers, the intuition behind SampleSkip holds. 
	
	\begin{definition}[Linear classifier]
		\label{def:linear}
		A linear classifier $f:R^d \to R^c$ is of the form $f_{\theta}(x) = \eta(\theta^w x + \theta^b)$, where model $\theta=(\theta^w \big \bracevert \theta^b)$ consists of a $c\times d$ matrix $\theta^w$ of weights and a $c$-sized vector $\theta^b$ of biases, and $\eta: R^c \to R^c$ can be any entrywise strictly increasing function: given any $c$-sized vector $v$, for any $j_1, j_2 \in [c]$, if $v[j_1] > v[j_2]$, we have $\eta(v)[j_1] > \eta(v)[j_2]$.
	\end{definition}
	\begin{theorem}
		\label{thm:sampleskip}
		For any linear classifier $f$, any test sample $(x, y)$, and any two models $\theta_{S_1}, \theta_{S_2}$, if they satisfy $\argmax_{j} f_{\theta_{S_1}}(x)[j] = \argmax_{j} f_{\theta_{S_2}}(x)[j]=y$, then for any $\theta_{S_1 \cup S_2} = \omega_{S_1} \cdot \theta_{S_1} + \omega_{S_2} \cdot \theta_{S_2}$ where $\omega_{S_1}, \omega_{S_2} \geq 0$, we have $\argmax_{j} f_{\theta_{S_1 \cup S_2}}(x)[j] = y$.
	\end{theorem}
	
	For nonlinear classifiers, this intuition is almost correct according to our experiments, and Lemma \ref{lemma:sampleskip} demonstrates the impact of the incorrectness on the estimated FSVs.
	Let $\hat{v}(\theta^t_{S})$ denote the utility of aggregated model $\theta^t_{S}$ estimated under SampleSkip.
	Then, for all aggregated models $\theta^t_{S}$, the incorrectness of SampleSkip can be measured by the percentage $\Delta v^t_{S}= \hat{v}(\theta^t_{S}) - v(\theta^t_{S})$ of wrongly skipped test samples, which is upper bounded by $\Delta v_{max}$.
	Lemma \ref{lemma:sampleskip} implies that $\Delta v_{max}$ affects the error $|\hat{\phi}_i - \phi_i|$ of the estimated FSV $\hat{\phi}_i$, and clients can infer the upper bound of the error by estimating $\Delta v_{max}$ with a small test set.
	Note that this upper bound is loose, and the exact error might be considerably smaller than it.
	\begin{lemma}
		\changemarker{
			\label{lemma:sampleskip}
			For each client $i$, the error $|\hat{\phi}_i - \phi_i|$ of her FSV $\hat{\phi}_i$ estimated under SampleSkip is upper bounded by $T\cdot\Delta v_{max}$. 
		}
	\end{lemma}
	\subsubsection{Security}
	\label{sec:security}
	Ensuring that no private information can be learned from our system is impossible because the FSVs themselves are knowledge of clients' private models and test data.
	Conversely, we focus on some well-known attacks. 
	A recent survey \cite{lyu2022privacy} suggests that the test phase of FL is threatened by model stealing and membership inference.
	Hence, we analyze the security of our systems against the membership inference attacks (simply \textit{MI attacks} \cite{shokri2017membership, hu2022membership}), the equation-solving model stealing attacks (simply \textit{ES attacks} \cite{tramer2016stealing}), and the retraining-based model stealing attacks (simply \textit{retraining attacks}, e.g., \cite{tramer2016stealing, papernot2017practical, orekondy2019knockoff, juuti2019prada}).
	The ES attacks can also be adapted for stealing test data (see Definition \ref{def:es_attack}).
	\begin{definition}[Equation-solving attack]
		\label{def:es_attack}
		Let $f^{(l_1:l_2)}$ denote the inference function that given input features for layer $l_1$ yields output features for layer $l_2$, where $l_1 \leq l_2$.
		Given input features $X^{(l_1)}$ for layer $l_1$ and output features $\hat{Y}^{(l_2)}$ for layer $l_2$, an ES attack obtains the model parameters $\theta^{(l_1)},...,\theta^{(l_2)}$ of layers $l_1,...,l_2$ by solving a system of equations $\hat{Y}^{(l_2)} = f^{(l_1:l_2)}(X^{(l_1)};\{\theta^{(l)}\}_{l=l_1}^{l_2})$.
		We refer to it as a single-layer ES attack if $l_1=l_2$; otherwise, as a multi-layer ES attack.
		Given model parameters $\theta^{(1)},...,\theta^{(l_2)}$ for layers $1,...,l_2$ and output features $\hat{Y}^{(l_2)}$ for layer $l_2$, an ES data stealing attack learns the input features $X^{(1)}$ by solving $\hat{Y}^{(l_2)} = f^{(1:l_2)}(X^{(1)};\{\theta^{(l)}\}_{l=1}^{l_2})$.
	\end{definition}
	\begin{definition}[Membership inference attack]
		\label{def:mi_attack}
		Given a test sample $(x,y)$ and its predicted label $\hat{y} = \argmax_{j} f_{\theta}(x)[j]$, an MI attack determines whether $(x,y)$ was in the training data of model $\theta$.
	\end{definition}
	
	\begin{definition}[Retraining attack]
		\label{def:retrain_attack}
		Given the features $X$ of some test samples and the predicted labels $\hat{Y}=\mathbf{Argmax}(f_{\theta}(X))$, a retraining attack learns a clone model $f'$ such that $\mathbf{Argmax}(f'(X)) \approx \hat{Y}$.
	\end{definition}
	\begin{proposition}
		\label{prop:security}
		Assume that linear layers are one-way functions, i.e., for all PPT adversaries $\mathcal{A}$ and for all linear layers $l$, there is a negligible $\epsilon$ such that $Pr[\mathcal{A}(\hat{Y}^{(l)}) = \theta^{(l)}] \leq \epsilon$ and $Pr[\mathcal{A}(\hat{Y}^{(l)}) = X^{(l)}] \leq \epsilon$ over all possible $(\theta^{(l)}, X^{(l)})$, where $\hat{Y}^{(l)}\! = \!lin^{(l)}(\theta^{(l)}, X^{(l)})$.
		Under SecSV or HESV, if $n \geq 4$, all PPT adversaries cannot apply single-layer ES attacks, ES data stealing attacks, MI attacks, or retraining attacks;
		if $n \geq L + 2$, they cannot apply multi-layer ES attacks.
	\end{proposition}
	Proposition \ref{prop:security} characterizes different security levels of our protocols with different numbers of clients.
	Intuitively, to defend against single-layer ES attacks, for each layer $l$, HESV and SecSV should ensure that clients $i_{l}, i_{l+1}$ are different entities such that they cannot know both the input $X^{(l)}$ and output $\hat{Y}^{(l)}$ to infer model parameters $\theta^{(l)}$.
	Similarly, clients $i_1,...,i_{L+1}$ should not be the owner of the model under testing, which prevents ES data stealing attacks;
	client $i_{L+1}$ who obtains the predicted labels $\hat{Y}$ should be different from client $i_{1}$ (the owner of the batch of samples) such that they cannot apply MI and retraining attacks.
	Note that if we have sufficient clients, SecSV can batch samples from different clients together;
	in this case, client $i_{1}$ represents the owners of the batched samples.
	For multi-layer ES attacks, HESV and SecSV require at least $L + 2$ clients to participate, which may be infeasible when $L$ is large.
	However, because our clients are honest-but-curious, we can consider a lower security level where only $4$ clients are required by ignoring the multi-layer ES attack whose success relies on a large number of random queries over the inference function $f^{(l_1:l_2)}$ by an active attacker;
	we can also slightly modify HESV and SecSV by applying garbled circuits to evaluate activation functions securely \cite{juvekar2018gazelle}, which prevents the above attacks for any number of clients.
	
	\begin{table*}[ht]
		\small
		\caption{Running time and error in detail. \changemarker{Brackets are used to indicate results under parallelization.}}
		\label{tab:eval}
		\begin{tabular}{ccccccccccc}
			\hline
			\multirow{2}{*}{\begin{tabular}[c]{@{}c@{}}Dataset\\ (model)\end{tabular}}   & \multirow{2}{*}{Protocol} & \multicolumn{6}{c}{Running time (s)}                                        & \multirow{2}{*}{\begin{tabular}[c]{@{}c@{}}Speedup w.r.t.\\ HESV (par.)\end{tabular}} & \multirow{2}{*}{\begin{tabular}[c]{@{}c@{}}Slowdown w.r.t.\\ NSSV (par.)\end{tabular}} & \multirow{2}{*}{Error} \\ \cline{3-8}
			&                           & Total (par.)       & Arithmetic & Enc.    & Dec.    & Comm.   & Shares gen. &                                                                                       &                                                                                        &                        \\ \hline
			\multirow{3}{*}{\begin{tabular}[c]{@{}c@{}}AGNEWS\\ (LOGI)\end{tabular}} & HESV                      & $4400$ ($460$)     & $3860$     & $121$   & $50.48$ & $96.25$ & $0.00$      & $1\times$ ($1\times$)                                                                 & $28.3\times$ ($29.0\times$)                                                            & $2.84\times 10^{-3}$   \\
			& SecSV                     & $1054$ ($126$)     & $1005$     & $6.61$  & $1.77$  & $8.53$  & $7.84$      & $4.2\times$ ($3.6\times$)                                                             & $6.8\times$ ($8.0\times$)                                                              & $9.64\times 10^{-4}$   \\
			& SecretSV                  & $3698$ ($376$)     & $146$      & $0.00$  & $0.00$  & $256$   & $1713$      & $1.2\times$ ($1.2 \times$)                                                            & $23.8\times$ ($23.7\times$)                                                            & $2.46\times 10^{-3}$   \\ \hline
			\multirow{3}{*}{\begin{tabular}[c]{@{}c@{}}BANK\\ (LOGI)\end{tabular}}   & HESV                      & $2934$ ($302$)     & $2255$     & $149$   & $64.99$ & $124$   & $0.00$      & $1\times$ ($1\times$)                                                                 & $13.5\times$ ($12.5\times$)                                                            & $2.13\times 10^{-3}$   \\
			& SecSV                     & $137$ ($18.75$)    & $113$      & $1.08$  & $1.14$  & $4.06$  & $7.81$      & $21.4\times$ ($16.1\times$)                                                           & $0.6\times$ ($0.8\times$)                                                              & $8.86\times 10^{-4}$   \\
			& SecretSV                  & $791$ ($81.04$)    & $26.64$    & $0.00$  & $0.00$  & $55.17$ & $376$       & $3.7\times$ ($3.7 \times$)                                                            & $3.6\times$ ($3.3 \times$)                                                             & $1.00\times 10^{-3}$   \\ \hline
			\multirow{3}{*}{\begin{tabular}[c]{@{}c@{}}MNIST\\ (CNN)\end{tabular}}       & HESV                      & $274470$ ($27600$) & $76384$    & $87610$ & $290$   & $29658$ & $0.00$      & $1\times$ ($1\times$)                                                                 & $276.9\times$ ($274.2\times$)                                                          & $6.98\times 10^{-4}$   \\
			& SecSV                     & $39310$ ($3992$)   & $32468$    & $11.27$ & $199$   & $592$   & $2299$      & $7.0\times$ ($6.9\times$)                                                             & $39.7\times$ ($39.7\times$)                                                            & $9.28\times 10^{-4}$   \\
			& SecretSV                  & $158407$ ($16003$) & $3751$     & $0.00$  & $0.00$  & $8933$  & $121393$    & $1.7\times$ ($1.7\times$)                                                             & $159.8\times$ ($159.0\times$)                                                          & $1.20\times 10^{-3}$   \\ \hline
			\multirow{3}{*}{\begin{tabular}[c]{@{}c@{}}miRNA-mRNA\\ (RNN)\end{tabular}}  & HESV                      & $411628$ ($41691$) & $190190$   & $61528$ & $738$   & $21481$ & $0.00$      & $1\times$ ($1\times$)                                                                 & $47.9\times$ ($48.4\times$)                                                            & $1.40\times 10^{-2}$   \\
			& SecSV                     & $77081$ ($7825$)   & $30215$    & $656$   & $321$   & $1288$  & $3050$      & $5.3\times$ ($5.3\times$)                                                             & $9.0\times$ ($9.1\times$)                                                              & $1.70\times 10^{-2}$   \\
			& SecretSV                  & $38567$ ($3902$)   & $1935$     & $0.00$  & $0.00$  & $2796$  & $22346$     & $10.7\times$ ($10.7 \times$)                                                          & $4.5\times$ ($4.5\times$)                                                              & $1.87\times 10^{-0}$   \\ \hline
		\end{tabular}
	\end{table*}
	
	\begin{table}[ht]
		\small
		\vspace{-8pt}
		\caption{Comparisons of SV estimation methods under SecSV.}
		\label{tab:sv_esti}
		\begin{tabular}{@{}cccl|cl@{}}
			\toprule
			\multirow{2}{*}{\begin{tabular}[c]{@{}c@{}}Dataset\\ (model)\end{tabular}}   & \multirow{2}{*}{Method}       & \multicolumn{2}{c|}{Speedup w.r.t. HESV}                                          & \multicolumn{2}{c}{Error ($\times 10^{-2}$)}                                     \\ \cmidrule(l){3-6} 
			&                               & \multicolumn{2}{c|}{\begin{tabular}[c]{@{}c@{}}SampleSkip \\ off/on\end{tabular}} & \multicolumn{2}{c}{\begin{tabular}[c]{@{}c@{}}SampleSkip \\ off/on\end{tabular}} \\ \midrule
			\multirow{4}{*}{\begin{tabular}[c]{@{}c@{}}AGNEWS\\ (LOGI)\end{tabular}} & \multicolumn{1}{c|}{SecSV}    & \multicolumn{2}{c|}{$4.2\times\quad 7.2\times$}                                   & \multicolumn{2}{c}{$0.10\quad 0.10$}                                             \\
			& \multicolumn{1}{c|}{SecSV+PS} & \multicolumn{2}{c|}{$4.2\times\quad 7.2\times$}                                   & \multicolumn{2}{c}{$2.00\quad 2.01$}                                             \\
			& \multicolumn{1}{c|}{SecSV+GT} & \multicolumn{2}{c|}{$3.5\times\quad 5.5\times$}                                   & \multicolumn{2}{c}{$3.41\quad 3.39$}                                             \\
			& \multicolumn{1}{c|}{SecSV+KS} & \multicolumn{2}{c|}{$5.3\times\quad 8.6\times$}                                   & \multicolumn{2}{c}{$17.63\quad 17.63$}                                           \\ \midrule
			\multirow{4}{*}{\begin{tabular}[c]{@{}c@{}}BANK\\ (LOGI)\end{tabular}}   & \multicolumn{1}{c|}{SecSV}    & \multicolumn{2}{c|}{$21.4\times\quad 36.6\times$}                                 & \multicolumn{2}{c}{$0.09\quad 0.09$}                                             \\
			& \multicolumn{1}{c|}{SecSV+PS} & \multicolumn{2}{c|}{$21.3\times\quad 36.5\times$}                                 & \multicolumn{2}{c}{$1.25\quad 1.24$}                                             \\
			& \multicolumn{1}{c|}{SecSV+GT} & \multicolumn{2}{c|}{$\hspace{0.1cm}8.9\times\quad 10.8\times$}                                 & \multicolumn{2}{c}{$3.40\quad 3.40$}                                             \\
			& \multicolumn{1}{c|}{SecSV+KS} & \multicolumn{2}{c|}{$27.0\times\quad 44.1\times$}                                 & \multicolumn{2}{c}{$7.67\quad 7.66$}                                             \\ \midrule
			\multirow{4}{*}{\begin{tabular}[c]{@{}c@{}}MNIST\\ (CNN)\end{tabular}}       & \multicolumn{1}{c|}{SecSV}    & \multicolumn{2}{c|}{$7.0\times\quad 25.8\times$}                                  & \multicolumn{2}{c}{$0.09\quad 0.64$}                                             \\
			& \multicolumn{1}{c|}{SecSV+PS} & \multicolumn{2}{c|}{$7.0\times\quad 25.8\times$}                                  & \multicolumn{2}{c}{$2.69\quad 2.88$}                                             \\
			& \multicolumn{1}{c|}{SecSV+GT} & \multicolumn{2}{c|}{$6.9\times\quad 25.3\times$}                                  & \multicolumn{2}{c}{$3.58\quad 3.80$}                                             \\
			& \multicolumn{1}{c|}{SecSV+KS} & \multicolumn{2}{c|}{$9.0\times\quad 27.2\times$}                                  & \multicolumn{2}{c}{$15.46\quad 15.65$}                                           \\ \midrule
			\multirow{4}{*}{\begin{tabular}[c]{@{}c@{}}miRNA-mRNA\\ (RNN)\end{tabular}}  & \multicolumn{1}{c|}{SecSV}    & \multicolumn{2}{c|}{$5.3\times\quad 11.8\times$}                                  & \multicolumn{2}{c}{$1.70\quad 1.82$}                                             \\
			& \multicolumn{1}{c|}{SecSV+PS} & \multicolumn{2}{c|}{$5.3\times\quad 11.8\times$}                                  & \multicolumn{2}{c}{$3.03\quad 3.25$}                                             \\
			& \multicolumn{1}{c|}{SecSV+GT} & \multicolumn{2}{c|}{$5.3\times\quad 11.7\times$}                                  & \multicolumn{2}{c}{$3.67\quad 3.50$}                                             \\
			& \multicolumn{1}{c|}{SecSV+KS} & \multicolumn{2}{c|}{$7.0\times\quad 14.0\times$}                                  & \multicolumn{2}{c}{$20.77\quad 20.49$}                                           \\ \bottomrule
		\end{tabular}
		 \vspace{-8pt}
	\end{table}
	\subsection{Discussion}
	
	\subsubsection{Connection with secure federated training}
	The protocol for secure SV calculation is generic and independent of the protocol for secure federated training, as the local models are assumed to be given.
	Thus, the training method is out of the scope of this study.
	Additionally, owing to the use of the FSV, secure SV calculation can be inserted into secure federated training, giving the server flexibility to decide when to perform contribution evaluation.
	For instance, the server may promptly calculate SSVs after each training round to identify clients with negative contributions at an early stage;
	if the training task is urgent, the clients can prioritize finishing all the training rounds over secure SV calculation.
	
	\subsubsection{Parallelization}
	There are two basic strategies for parallelizing secure SV calculation.
	First, because the FSV is the sum of multiple rounds of SSVs, which are independent of each other, the server(s) can recruit a process for each round of SSV calculation to parallelize FSV calculation.
	Second, the server(s) can further parallelize each round of secure model testing by distributing the numerous models or test samples to multiple processes for evaluation.
	However, this strategy cannot be easily adopted for SampleSkip, under which the workload of testing a model may depend on the evaluation result of another model.
	Developing parallelization methods that are compatible with SampleSkip would be interesting.
	
	\subsubsection{Accelerating SV calculation}
	To calculate SSVs in a single round of FL, the server has to evaluate $O(2^n)$ models on $M$ test samples, which requires testing models $O(2^n\cdot M)$ times.
	Therefore, we can reduce the scale of (1) models or (2) samples to be tested.
	The existing SV estimation methods \cite{fatima2008linear, maleki2015addressing, ghorbani2019data, jia2019towards} use the first strategy:
	they sample and test only a subset of models for SV calculation.
	Skipping a portion of models does not introduce a significant estimation error because these methods assume that numerous clients participate in collaborative ML.
	However, in cross-silo FL, the number $n$ of clients is relatively small, causing the failure of the above methods in skipping models.
	SampleSkip instead adopts the second strategy and is effective in those cases where model testing is a significant bottleneck, e.g., when HE is used.
	Moreover, we can use both strategies and easily combine SampleSkip with an existing model-skipping estimation method. 
	Note that although $n$ is small and our task is parallelizable, accelerating secure SV calculation is crucial because HE slows down SV calculation drastically.
	
	\subsubsection{Dynamic scenarios}
	Our methods also work when clients join in or leave FL halfway.
	For each client, we only need to calculate SSVs for the rounds she attended and add them up as her FSV.

	\subsubsection{Other variants of SV}
	Our techniques also support other FL-oriented variants of SV, including the Contribution Index \cite{song2019profit}, Group SV \cite{ma2021transparent}, GTG-Shapley \cite{liu2022gtg}, and WT-Shapley \cite{yang2022wtdp}.
	Intuitively, these metrics are functions of the utilities of local and aggregate models;
	thus, after securely testing model utilities by our methods, we can easily calculate them.
	This characterization also suggests that our methods could be adapted to other model utility-based contribution metrics for FL.

	\section{Experiments}

	\subsection{Setup}
	\subsubsection*{Research questions} 
	We experiment to answer these questions.
	\begin{itemize}[leftmargin=*]
		\item \textit{Evaluation efficiency (RQ1)}: 
		How efficient are SecSV and HESV for secure SV calculation? 
		How much can SampleSkip accelerate SecSV?
		How efficient are Matrix Squaring and Matrix Reducing for secure matrix multiplication?
		\item \textit{Evaluation accuracy (RQ2)}: 
		How accurately do SecSV and HESV calculate FSVs?
		How many test samples are wrongly identified as skippable samples by SampleSkip?
		\item \textit{Parameters' effects (RQ3)}: 
		What are the effects of the size $M$ of test samples, number $n$ of clients, and number $L$ of layers on evaluation efficiency and accuracy?
	\end{itemize}

	\subsubsection*{Baselines}
	For secure SV calculation, we compare SecSV with HESV and Nonsecure SV (NSSV), which calculates exact FSVs in a nonsecure environment.
	In addition, we design a two-server baseline protocol named \textit{SecretSV} for comparison, which protects models and test data purely by ASS.
	See \venueforappendix{\ref{appendix:protocol}} to check the details of NSSV and SecretSV.
	For SV estimation, we compare SampleSkip with three widely adopted methods: \textit{Permutation Samples (PS)} \cite{maleki2015addressing}, \textit{Group Testing (GT)} \cite{jia2019towards}, and \changemarker{\textit{KernelSHAP (KS)} \cite{lundberg2017kernelshap}}.
	
	\begin{table}[t]
		\small
		\vspace{-8pt}
		\caption{Speedup of Matrix Reducing w.r.t. Matrix Squaring in the time per sample spent on HE computations for evaluating $AB$. The shape of matrix $A$ is varied. "Full" means both $A$ and $B$ are encrypted, whilst "Half" means only $A$ is encrypted.}
		\label{tab:matmul_eval}
		\begin{tabular}{@{}c|ccccccc@{}}
			\toprule
			Shape & $4\!\times\!300$ & $2\!\times\!48$ & $64\!\times\!256$ & $10\!\times\!64$ & $32\!\times\!64$ & $32\!\times\!32$ & $2\!\times\!32$ \\ \midrule
			Full  & $1.69\times$     & $6.10\times$    & $1.99\times$      & $2.30\times$     & $2.66\times$     & $2.85\times$     & $2.45\times$    \\
			Half  & $3.24\times$     & $11.39\times$   & $3.92\times$      & $4.49\times$     & $5.23\times$     & $3.71\times$     & $2.87\times$    \\ \bottomrule
		\end{tabular}
		 \vspace{-8pt}
	\end{table}

	\begin{figure*}[ht]
		\centering
		\begin{minipage}[t]{0.25\linewidth}
			\centering
			\includegraphics[scale=0.29]{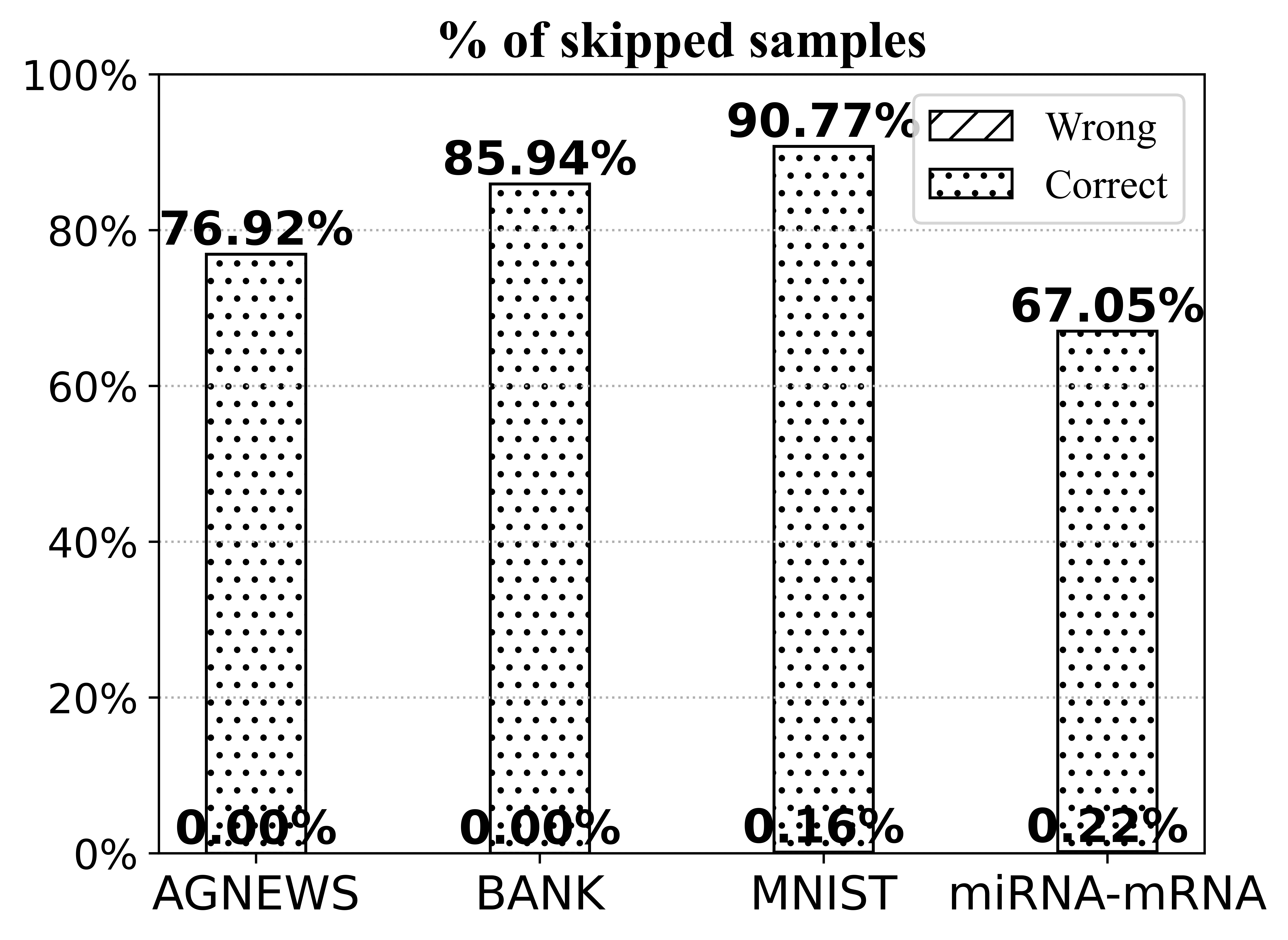}
			\vspace{-17pt}
			\caption{Skipped samples.}
			\label{fig:skipped}
		\end{minipage}%
		\begin{minipage}[t]{0.25\linewidth}
			\centering
			\includegraphics[scale=0.28]{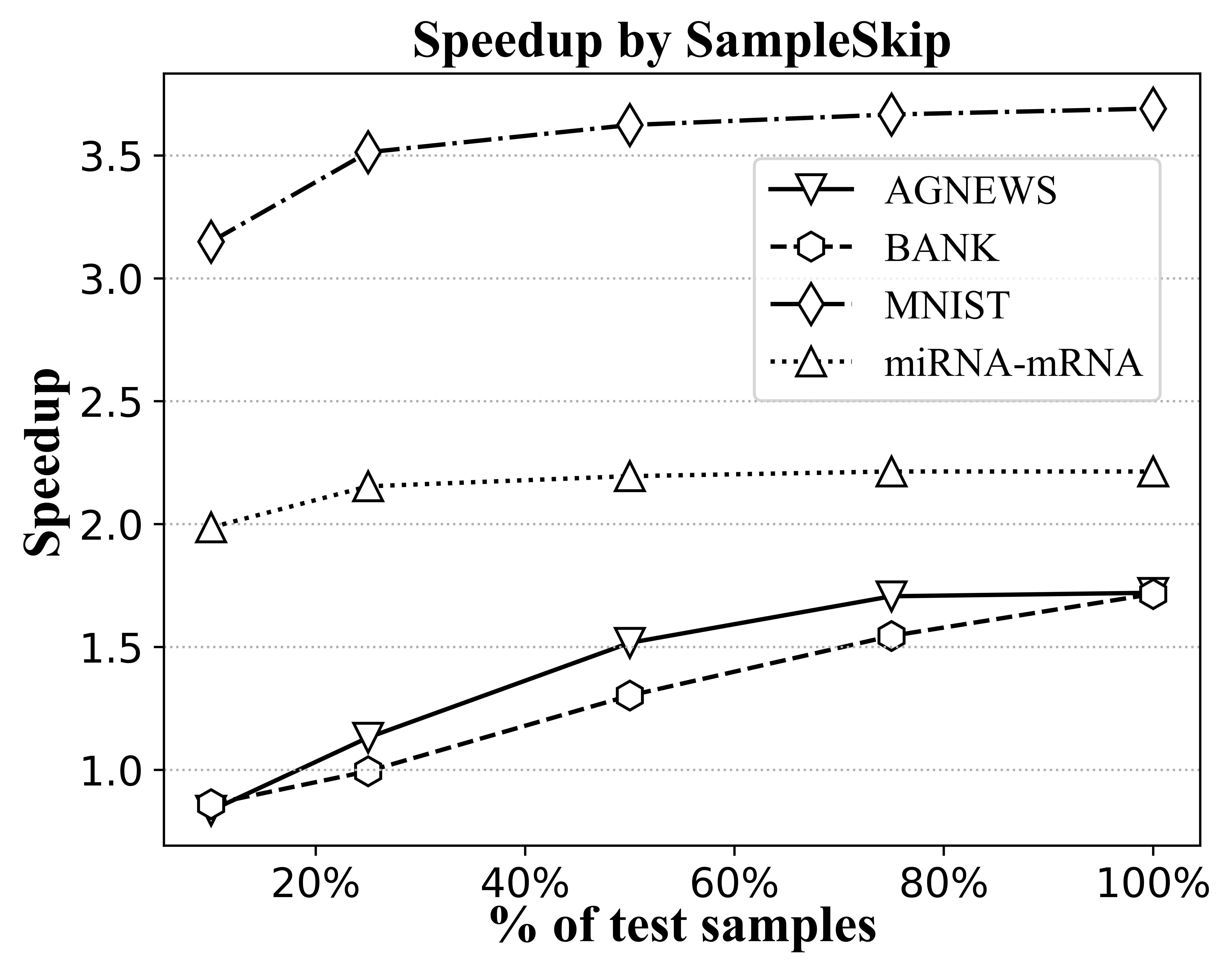}
			\vspace{-6pt}
			\caption{Effect of $M$.}
			\label{fig:vary_M}
		\end{minipage}%
		\begin{minipage}[t]{0.25\linewidth}
			\centering
			\includegraphics[scale=0.28]{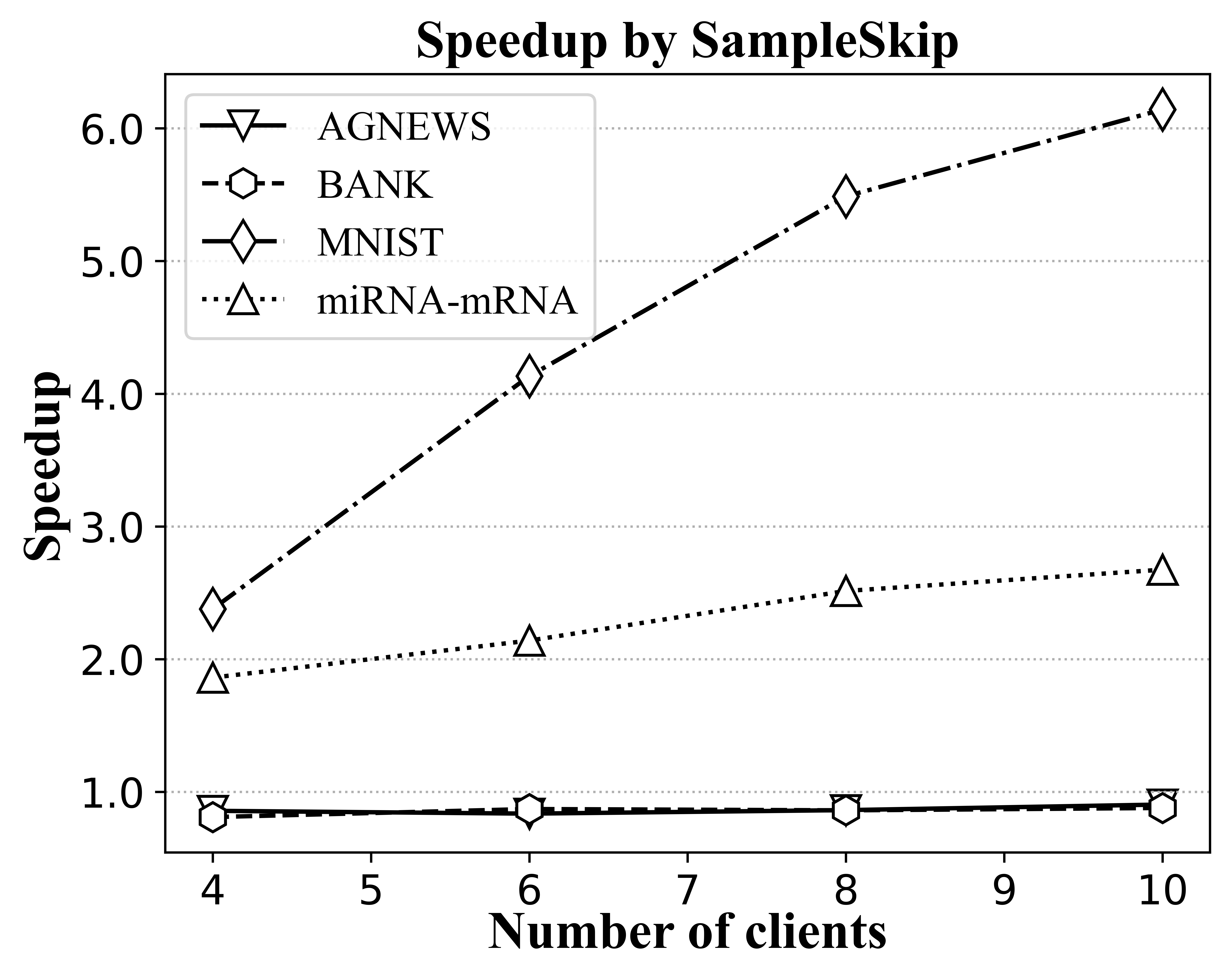}
			\vspace{-6pt}
			\caption{\changemarker{Effect of $n$.}}
			\label{fig:vary_n}
		\end{minipage}%
		\begin{minipage}[t]{0.25\linewidth}
			\centering
			\includegraphics[scale=0.28]{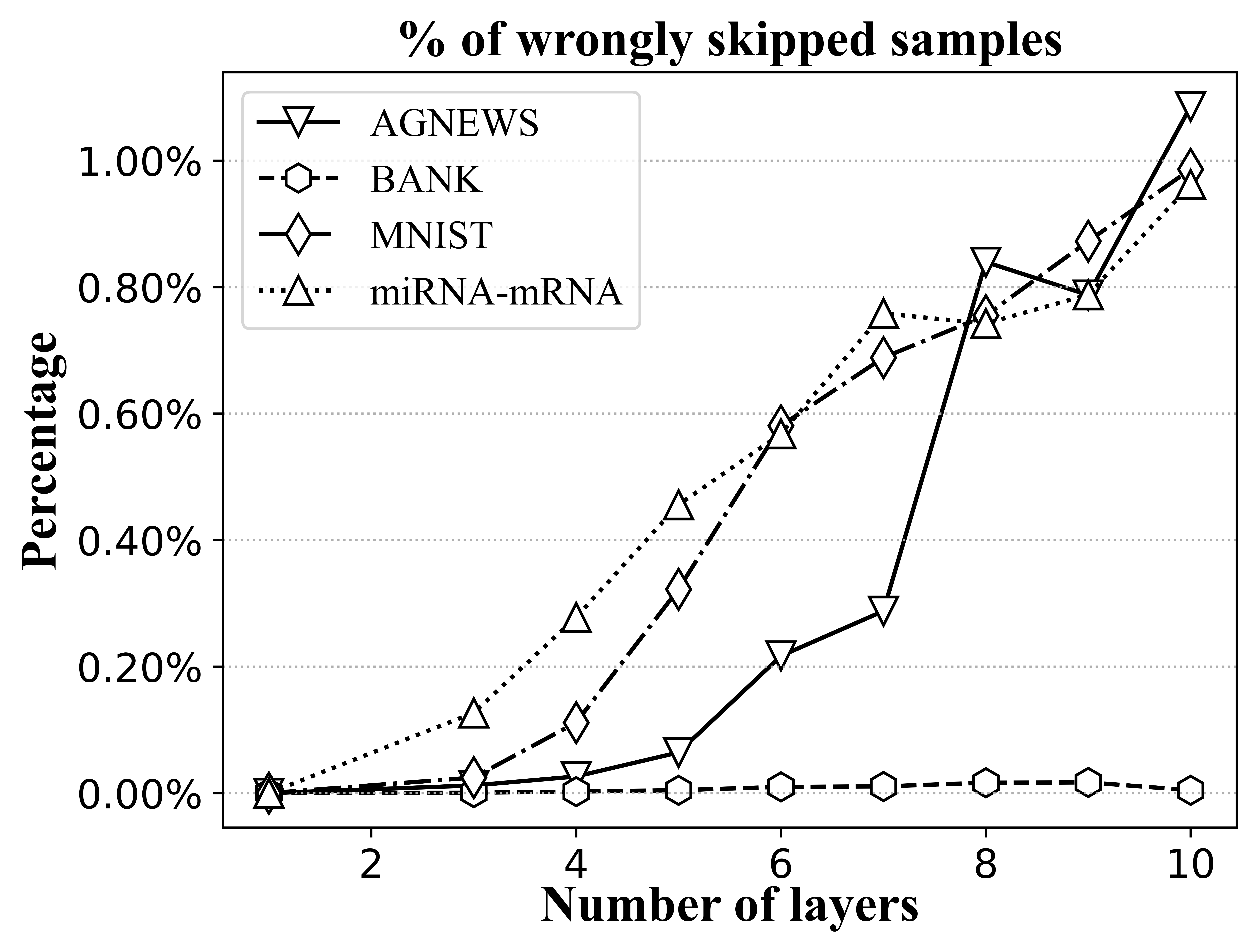}
			\vspace{-6pt}
			\caption{\changemarker{Effect of $L$.}}
			\label{fig:vary_L}
		\end{minipage}%
		\vspace{-8pt}
	\end{figure*}
	
	\subsubsection*{Datasets and classifiers} 
	We use the \textit{AGNEWS} \cite{zhang2015character}, \textit{BANK} \cite{moro2014data}, \textit{MNIST} \cite{lecun1998gradient}, and \textit{miRNA-mRNA} \cite{menor2014mirmark} datasets, which cover the news classification, bank marketing, image recognition, and miRNA target prediction tasks, for experiments.
	By default, we consider logistic classifiers (LOGI) for AGNEWS and BANK, a convolutional NN (CNN) \cite{jiang2018e2dm} for MNIST, and a recurrent NN (RNN) \cite{lee2016deeptarget} for miRNA-mRNA;
	when evaluating the effect of the number $L$ of layers, we rather focus on the deep NN (DNN) to easily vary $L$, which consists of an input layer, an output layer, and $L - 2$ hidden layers.
	See \venueforappendix{\ref{appendix:data_model}} for details of the datasets and classifiers.
	
	\subsubsection*{Federated training} 
	We  allocate each label of training and test samples to clients according to the Dirichlet distribution of order $n$ with scaling parameters of $0.5$ to construct non-IID local data, following \cite{yurochkin2019bayesian}.
	Subsequently, $n=5$ clients perform $T=10$ rounds of the FedAvg algorithm \cite{mcmahan2016federated} to train the above classifiers.

	\subsubsection*{Evaluation metrics}
	The computation efficiency of a secure SV calculation protocol is measured by its speedup w.r.t. the baseline HESV.
	Subsequently, following \cite{jia2019towards}, we use the Euclidean distance to measure the error of the estimated FSVs.
	Concretely, let $\boldsymbol{\hat{\phi}}$ and $\boldsymbol{\phi}$ denote the vectors of the estimated FSVs and the exact FSVs for all the clients. 
	The (expected) error of $\boldsymbol{\hat{\phi}}$ is defined as $\mathbb{E}_{\boldsymbol{\phi}}[||\boldsymbol{\hat{\phi}} - \boldsymbol{\phi}||_2]$.
	Note that both HE and ASS introduce noise into messages, and the exact FSVs are calculated by the nonsecure protocol NSSV.
	
	\subsubsection*{Environment}
	We implement the server and clients on a workstation with a $3.0$ GHz 16-core processor and $128$ GB RAM and logically simulate $1$ Gbps communication channels between the parties.  
	When evaluating efficiency under parallelization, we allocate a dedicated CPU core for each round of SSV calculation.
	We use TenSEAL \cite{tenseal2021}, a Python wrapper for Microsoft's SEAL library, to implement HE operations and select the parameters of HE and ASS for $128$-bit security based on the \textit{Homomorphic Encryption Security Standard} \cite{HomomorphicEncryptionSecurityStandard}, which results in $N=2048$ ciphertext slots.
	We run each experiment $10$ times and present the average results.
	
	\subsection{Experimental Results}

	\subsubsection{Evaluation efficiency (RQ1)}
	Table \ref{tab:eval} shows the running time of different protocols spent on secure SV calculation.
	We can see that HESV requires a considerably long running time even under parallelization and is dramatically less efficient than NSSV because the use of HE is extremely expensive in computation. This verifies the importance of SV estimation in our setting.
	The slowdown effect could be more significant if we run NSSV on GPUs.
	SecSV spends considerably less time than HESV on HE computations owing to the avoidance of c2c multiplications; 
	the time it takes to generate secret shares is insignificant compared with the total running time.
	Although SecSV requires communication with an auxiliary server, it may consume less communication time than HESV because the size of a message's secret shares is significantly smaller than that of its ciphertext.
	SecretSV performs arithmetic operations efficiently as models and test data are not encrypted.
	However, shares generation and communication are time-consuming because a triplet of random masks should be secretly shared between the two servers for every multiplication under ASS. 
	
	Table \ref{tab:sv_esti} compares SampleSkip with the existing SV estimation methods in terms of acceleration.
	Whereas SampleSkip and KS significantly speed up SecSV, the PS and GT methods hardly accelerate it because they fail to skip any models.
	For the logistic classifiers, GT significantly slows down SecSV because it needs to calculate FSVs by solving an optimization problem, which is expensive for light classifiers. 
	Furthermore, SampleSkip can be combined with all these baselines to significantly improve efficiency.
	This marvelous performance is caused by the high percentage of test samples skipped by SampleSkip, as shown in Figure \ref{fig:skipped}.
	
	Finally, we evaluate the efficiency of Matrix Squaring and Matrix Reducing.
	As seen in Table \ref{tab:matmul_eval}, we evaluate secure matrix multiplications between a $d_{out} \times d_{in}$ matrix $A$ and a $d_{in} \times m$ matrix $B$;
	we set batch size $m=\lfloor N/d_{out} \rfloor$ for Matrix Reducing and $m=\min\{d_{in},\lfloor \sqrt{N} \rfloor\}$ for Matrix Squaring.
	Subsequently, we enumerate all the matrix multiplications involved in the four classifiers considered in our experiments and present the speedup of Matrix Reducing w.r.t. Matrix Squaring in each case.
	When numerous homomorphic rotations are required for Matrix Squaring, e.g., when $A$ is of shapes  $2 \times 48$, the computation time per sample under Matrix Reducing is considerably less than that under Matrix Squaring.
	Additionally, compared with the case where both $A$ and $B$ are encrypted, when only $A$ is encrypted, the speedup of Matrix Reducing is higher because the computation time required for homomorphic multiplications becomes less whereas that required for homomorphic rotations remains the same.

	\subsubsection{Evaluation accuracy (RQ2)}
	As shown in Table \ref{tab:eval}, both SecSV and HESV can accurately calculate FSVs for all datasets.
	The error of the FSVs calculated by SecretSV is passable for the logistic and CNN classifiers but overly large for RNN.
	This may be because RNN has numerous linear layers in sequence, through which ASS introduces accumulated truncation errors into prediction results.
	
	Table \ref{tab:sv_esti} suggests that SampleSkip introduces only an insignificant error for SecSV even when combined with other estimation methods.
	The PS and GT methods vastly increase the error because they approximate the FSVs even if they skip no models.
	The KS method skips some models and accelerates SecSV to some extent but also results in an unacceptable error for all the datasets because it works poorly when the number $n$ of clients is small.
	
	We also present the percentages of the samples skipped by SampleSkip in Figure \ref{fig:skipped}.
	It is evident that no samples are wrongly skipped for the linear logistic classifiers, which experimentally demonstrates the correctness of Theorem \ref{thm:sampleskip}.
	For the nonlinear CNN and RNN classifiers, only a tiny portion of samples are wrongly skipped, which results in a minor error.
	
	\subsubsection{Parameters' effects (RQ3)}
	First, we vary the size $M$ of test samples to evaluate its effect on the efficiency of SecSV.
	As depicted in Figure \ref{fig:vary_M}, when more test samples are used for evaluation, SampleSkip becomes more effective in accelerating SecSV.
	This effect may be due to two reasons.
	First, SecSV usually requires a large batch size to perform at its best owing to the use of Matrix Reducing.
	When the size is small, and the majority of the samples are skipped, the remaining discriminative samples might be overly few to fulfill the batch, which causes a waste of ciphertext slots.
	Second, SampleSkip needs to preprocess discriminative samples for each aggregated model, which is relatively time-consuming, particularly for light classifiers.
	Thus, for the BANK and AGNEWS datasets with logistic models, when only a few samples are evaluated, SampleSkip may slow down SecSV.
	Note that we can easily avoid this problem by adaptively turning off SampleSkip.
	
	We then vary the number $n$ of clients to test its effect on the efficiency of SecSV.
	In this case, we use $10\%$ of the entire test samples.
	In Figure \ref{fig:vary_n}, we can see that when $n$ increases, SampleSkip accelerates SecSV more for MNIST and miRNA-mRNA but probably less for AGNEWS and BANK.
	The reason for this effect is as follows.
	An increase of $n$ implies a higher percentage of aggregated models w.r.t. all models because the number of aggregated models increases exponentially w.r.t. $n$.
	The increase in the percentage of aggregated models can amplify the speedup and slowdown effects of SampleSkip.
	Therefore, because SampleSkip speeds up SecSV for MNIST and miRNA-mRNA in Figure \ref{fig:vary_M}, a larger $n$ results in a higher speedup.
	However, because SampleSkip slows down SecSV when only $10\%$ of the test samples are used owing to the need to preprocess discriminative samples, a larger $n$ may cause a higher slowdown.
	Note that the lines for AGNEWS and BANK in Figure \ref{fig:vary_n} may rise because more models need to be securely tested; this makes the speedup approach $1.0$. 
	
	Finally, we vary the number $L$ of layers in the DNN classifier.
	Figure \ref{fig:vary_L} shows that when $L=1$, which means DNN is reduced to a logistic classifier because no hidden layer exists, no samples are wrongly skipped, which verifies Theorem \ref{thm:sampleskip}.
	When $L$ increases, DNN becomes "more nonlinear"; thus, more samples may be wrongly skipped. 
	However, the percentages of the wrongly skipped samples remain small even when $L=10$.
	
	\section{Related Work}
	\subsubsection*{Secure FL}
	There are two typical FL settings: cross-device and cross-silo \cite{kairouz2021advances}.
	In the cross-device setting, the clients are numerous unreliable mobile or IoT devices with limited computation and communication resources, making using lightweight cryptographic techniques to protect clients' privacy reasonable.
	Therefore, existing secure cross-device FL systems \cite{bonawitz2017practical, guo2020v, bell2020secure, choi2020communication, kadhe2020fastsecagg, so2020byzantine, so2021turbo, so2022lightsecagg} usually protect privacy using secure aggregation:
	they mask the clients' updates using secret sharing, and the mask for each update will be canceled by aggregating the masked updates. 
	Nevertheless, secure aggregation can only protect the local updates; the final model is still leaked to the FL server or a third party.
	Some studies \cite{geyer2017differentially, mcmahan2017learning, triastcyn2019federated, hu2020personalized, zhao2020local, kairouz2021distributed, girgis2021shuffled, liu2021flame, noble2022differentially} protect privacy against a third party using differential privacy (DP) \cite{dwork2006calibrating} or even against the server using local DP \cite{evfimievski2003limiting}.
	However, the uses of DP and local DP significantly compromise the final model's utility.
	Conversely, the clients in cross-silo FL are a small number of organizations with abundant communication and computing resources.
	Considering that the purpose of the clients is to train an accurate model for their own use, they might not allow releasing the final model to external parties, including the server, and/or tolerate a nontrivial loss of model utility \cite{zhang2020batchcrypt}.
	Consequently,  HE is a natural choice for cross-silo FL, which is advocated by many works \cite{phong2018privacy, truex2019hybrid, zhang2020batchcrypt, sav2020poseidon, zhang2021dubhe, jiang2021flashe, ma2022privacy}.
	Although HE largely increases the computation and communication overhead, it is acceptable to those powerful clients.
	However, the existing HE-based FL systems cannot support secure SV calculation because they only protect the model aggregation process, which inspires this paper.
	Additionally, Ma et al. \cite{ma2021transparent} studied a similar problem to ours; 
	however, they aimed to make the contribution evaluation process transparent and auditable using blockchain, whereas we focus on the security of this process.

	\subsubsection*{SV for collaborative ML}
	The SV has been widely adopted in collaborative ML for contribution evaluation and numerous downstream tasks.
	Concretely, because the value of data is task-specific, the contribution metric SV can be considered a data valuation metric \cite{ghorbani2019data, jia2019towards, wang2020principled, wei2020efficient}.
	Some studies \cite{song2019profit,ohrimenko2019collaborative,liu2020fedcoin, han2020replication} proposed SV-based payment schemes that allocate monetary payments to clients based on their respective SVs.
	The evaluated SVs can also be used to decide the qualities of models \cite{sim2020collaborative} or synthetic data \cite{tay2022incentivizing} that are rewarded to clients and to identify irrelevant clients who may have negative contributions to the final model \cite{nagalapatti2021game}.

	Given that SV calculation requires exponential times of model retraining and evaluation to accelerate SV calculation, various estimation methods \cite{fatima2008linear, maleki2015addressing, ghorbani2019data, jia2019towards, song2019profit, wei2020efficient, wang2020principled, liu2022gtg} were proposed, and they can be categorized into two: \textit{retraining-skipping} and \textit{model-skipping}.
	The former class of methods reduces the complexity of or even eliminates the redundant retraining process.
	For example, the FSV \cite{wang2020principled}, MR \cite{song2019profit}, Truncated MR \cite{wei2020efficient}, and GTG-Shapley \cite{liu2022gtg} methods use the local and aggregated models for SV calculation, which avoids retraining FL models.
	The model-skipping methods, e.g., the PS \cite{maleki2015addressing} and GT \cite{jia2019towards} methods used in our experiments, approximate the SV by sampling a sufficient number of models for evaluation.
	Our SampleSkip method skips test samples to accelerate the model evaluation process, which does not belong to the above classes and opens up a new direction: \textit{sample-skipping estimation}.

	\subsubsection*{Secure inference}
	This problem regards securely making predictions given a client's encrypted input features for prediction-as-a-service (PaaS).
	The majority of existing studies (e.g., \cite{gilad2016cryptonets, juvekar2018gazelle, mishra2020delphi, reagen2021cheetah, meftah2021doren}) assumed that the prediction model is not outsourced because the server of PaaS is usually the model owner himself and thus do not encrypt the model, which is different from our privacy model.
	Several studies \cite{mohassel2017secureml, jiang2018e2dm, hesamifard2018privacy} examined the scenario of \textit{secure outsourced inference} where the model is outsourced to and should be protected against an untrusted server, which is similar to our scenario.
	Nevertheless, in our problem, we need to ensure security against more untrusted parties, and the server(s) should securely evaluate numerous models on massive test samples, which calls for estimation methods to reduce the scale of the task.
	
	\section{Conclusion}
	In this paper, we introduced the problem of secure SV calculation for cross-silo FL and proposed two protocols.
	The one-server protocol HESV had a stronger security guarantee because it only required a single server to coordinate secure SV calculation, whereas the two-server protocol SecSV was considerably more efficient. 
	Solving this problem facilitates many downstream tasks and is an essential step toward building a trustworthy ecosystem for cross-silo FL.
	Future directions include studying secure SV calculation in the vertical FL setting, where clients possess different attributes of the same data samples, and developing secure protocols for other types of prediction models and more effective acceleration techniques.
	
	\begin{acks}
		We are grateful for the generous support provided by JSPS to the first author, a JSPS Research Fellow, as well as the support from JST CREST (No. JPMJCR21M2), JST SICORP (No. JPMJSC2107), and JSPS KAKENHI (No. 17H06099, 21J23090, 21K19767, 22H03595).
	\end{acks}
	
	
	\bibliographystyle{ACM-Reference-Format}
	\balance
	\bibliography{ref}

\showappendix{
	\begin{algorithm}[h]
		\small
		\caption{Nonsecure SV (NSSV)}
		\begin{algorithmic}[1]
			\STATE Each client $i$: Submit test data $D_i$ to Server
			\FOR{$t$ in $\{1,...,T\}$}
			\FOR {$S \subseteq I^t, |S| > 0$}
			\STATE Server: Calculate $\theta^t_S$ and evaluate its utility $v(\theta^t_S)$
			\ENDFOR
			\STATE Server $\mathcal{P}$: Calculate SSVs $\phi^{t}_1,...,\phi^{t}_n$
			\ENDFOR
			\STATE Server $\mathcal{P}$: Calculate FSVs $\phi_1,...,\phi_n$
			\RETURN $\phi_1,...,\phi_n$
		\end{algorithmic}
		\label{alg:nssv}
	\end{algorithm}

	\begin{algorithm}[h]
		\small
		\caption{SecretSV}
		\begin{algorithmic}[1]
			\STATE Server $\mathcal{P}$: Randomly select a leader client
			\STATE Leader: Generate a public key $pk$ and a private key $sk$ of HE and broadcast them among the other clients
			\STATE Each client $i$: Generate secret shares $\theta'^t_i, \theta''^t_i$ of $\theta^t_i$, send $\theta'^{t}_i$ to server $\mathcal{P}$, and send $\theta''^{t}_i$ to server $\mathcal{A}$.
			\STATE Each client $i$: Generate secret shares $D_i'=(X_i', Y_i'), D_i''=(X_i'', Y_i'')$ of $D_i=(X_i,Y_i)$, send $D_i'$ to server $\mathcal{P}$, and send $D_i''$ to server $\mathcal{A}$
			\FOR{$t$ in $\{1,...,T\}$}
			\FOR {$S \subseteq I^t, |S| > 0$}
			\STATE Servers $\mathcal{P}, \mathcal{A}$: Calculate $\theta'^t_S$ and $\theta''^t_S$ 
			\STATE Server $\mathcal{P}$: Run Algorithm \ref{alg:infer_ass-sv}  $\Pi_{ASS}((\theta'^t_S, \theta''^t_S), (\mathcal{D}', \mathcal{D}''))$ to obtain $cnt^t_{S}$
			\STATE Server $\mathcal{P}$: Calculate $v(\theta^t_S) = cnt^t_{S}/M$
			\ENDFOR
			\ENDFOR
			\STATE Server $\mathcal{P}$: Calculate SSVs $\phi^{t}_1,...,\phi^{t}_n, \forall t \in [T]$
			\STATE Server $\mathcal{P}$: Calculate FSVs $\phi_1,...,\phi_n$
			\RETURN $\phi_1,...,\phi_n$
		\end{algorithmic}
		\label{alg:secretsv}
	\end{algorithm}
	
	\begin{algorithm}[h]
		\small
		\caption{$\Pi_{ASS}$: Secure Testing for SecretSV}
		\begin{algorithmic}[1]
			\REQUIRE secret shares $\theta'^t, \theta''^t$ of model, and secret shares $\mathcal{D}', \mathcal{D}'$ of collective test set
			\ENSURE count $cnt$ of correct predictions 
			\STATE $cnt \gets 0$
			\FOR {each $D'=(X',Y'), D''=(X'',Y'') \in (\mathcal{D}', \mathcal{D}'')$}
			\FOR {each model layer $l\in \{1,...,L\}$}
			\STATE Server $\mathcal{A}$: Calculate $\hat{Y}''^{(l)} = lin^{(l)}(\theta''^t, X''^{(l)})$ and send $\hat{Y}''^{(l)}$ to client $i_{l+1} \in I \backslash\{i_{l}\}$, where $X''^{(1)} = X''$
			\STATE Server $\mathcal{P}$: Calculate $ \hat{Y}'^{(l)}  = lin^{(l)}(\theta'^t, X'^{(l)})$ and send $\hat{Y}'^{(l)}$ to client $i_{l+1} \in I \backslash\{i_{l}\}$, where $X'^{(1)} = X'$
			\IF{$l < L$}
			\STATE Client $i_{l+1}$: Compute $X^{(l+1)}= ac^{(l)}((\hat{Y}'^{(l)} + \hat{Y}''^{(l)}) \bmod p)$, generate shares $X'^{(l+1)}$ and $X''^{(l+1)}$, and send them to servers $\mathcal{P}$ and $\mathcal{C}$, respectively
			\ELSE 
			\STATE Client $i_{L+1}$: Calculate $\hat{Y}=\mathbf{Argmax}((\hat{Y}'^{(L)} + \hat{Y}''^{(L)})\bmod p)$, generate shares $\hat{Y}', \hat{Y}''$, and send them to servers $\mathcal{P}$ and $\mathcal{C}$, respectively
			\ENDIF
			\ENDFOR
			\STATE Server $\mathcal{A}$: Compute and send $\tilde{Y}''= \hat{Y}'' - Y''$ to server $\mathcal{P}$. 
			\STATE Server $\mathcal{P}$: Calculate $\tilde{Y} = abs((\hat{Y}' - Y' + \tilde{Y}'') \bmod p)$, where $abs$ denotes taking the entrywise absolute value.
			\STATE Server $\mathcal{P}$: $cnt \gets cnt + \sum_{k=1}^{m} \boldsymbol{1}(|\tilde{Y}[k]| < 0.5)$
			\ENDFOR
			\RETURN $cnt$
		\end{algorithmic}
		\label{alg:infer_ass-sv}
	\end{algorithm}
	
	\section{Missing Protocols}
	\label{appendix:protocol}
	The Nonsecure SV (NSSV) and SecretSV protocols are presented in Algorithms \ref{alg:nssv} and \ref{alg:secretsv}.
	NSSV just calculates the ground-truth FSVs in a nonsecure environment.
	SecretSV protects models and test data based purely on ASS, and its secure testing protocol (Algorithm \ref{alg:infer_ass-sv}) is adapted from a SOTA three-party ASS-based secure ML protocol \cite{wagh2019securenn}.
	The method for matrix multiplications between secret shares can be found in \cite{wagh2019securenn}.
	
	
	\section{Missing Proofs}
	\label{appendix:proof}
	
	\begin{proof}[Proof of Lemma \ref{lemma:correctness}]
		For all $j \in [d_{out}]$ and $k \in [m]$, we have
		\begin{align*}
			\delta(A, B)[j, k] &  = \sum_{o=1}^{d_{in}} \xi^{(o-1)}(\sigma(A))[j,k] \cdot \psi^{(o-1)}(\tau(B))[j,k] \\
			& = \sum_{o=1}^{d_{in}} \sigma(A)[j,k+o-1] \cdot \tau(B)[j+o-1,k] \\
			& = \sum_{o=1}^{d_{in}} A[j,j+k+o-1] \cdot B[j+k+o-1,k] \\
			& = \sum_{o=1}^{d_{in}} A[j,o] \cdot B[o,k] = (AB)[j, k]
		\end{align*}
	\end{proof}
	
	\begin{lemma}
		\label{lemma:mat_square}
		To evaluate $AB$ under HE with $d_{out} \leq d_{in}$ and $m \leq \min\{d_{in},\lfloor \sqrt{N} \rfloor\}$, Matrix Squaring needs $O(\frac{d_{in} \cdot d_{out}}{\sqrt{N}})$ $HMult$s and $O(\frac{d_{in}}{d_{out} \bmod \sqrt{N}})$ $HRot$s.
	\end{lemma}
	
	\begin{proof}[Proof of Lemma \ref{lemma:commutative}]
		Let $A'=(A\big\bracevert A)$, $B'=(B_1\big\bracevert B_2)$, and $B''=(B';B')$.
		By running the first two steps of $\delta(A', B')$ and $\delta(A', B'')$, we derive two sets of matrices $\{\tilde{A}'^{(o)}\}_{o=1}^{d_{in}}$ and $\{\tilde{B}'^{(o)}\}_{o=1}^{d_{in}}$ for $\delta(A', B')$ and $\{\tilde{A}'^{(o)}\}_{o=1}^{2d_{in}}$ and $\{\tilde{B}''^{(o)}\}_{o=1}^{2d_{in}}$ for $\delta(A', B'')$, respectively, where
		\begin{equation*}
			\forall o \in [2d_{in}], \tilde{A}'^{(o)} = \tilde{A}'^{((o-1) \bmod d_{in}+1)}, \tilde{B}''^{(o)} = \tilde{B}'^{((o-1) \bmod d_{in}+1)} 
		\end{equation*}
		Obviously, we have $2\delta(A', B') =\delta(A', B'')$ and can conclude that
		\begin{align*}
			&(\delta(A, B_1) \big\bracevert \delta(A, B_2)) = (AB_1 \big\bracevert AB_2) = AB' \\
			= &\frac{1}{2}A'B'' = \frac{1}{2}\delta(A', B'') = \delta(A', B') = \delta((A \big\bracevert A), (B_1 \big\bracevert B_2))
		\end{align*}
		
	\end{proof}
	
	\begin{proof}[Proof of Lemma \ref{lemma:mat_square}]
		When $d_{out} \leq d_{in} \leq \lfloor \sqrt{N} \rfloor$, Matrix Squaring, i.e., the SOTA method, pads $A$ with zeros to derive a $d'_{out} \times d_{in}$ matrix $A'$ such that $d'_{out}$ exactly divides $d_{in}$.
		In the best case where $d_{out}$ exactly divides $d_{in}$, we have $d'_{out}=d_{out}$, while in the worst case where no number $d \in (d_{out}, d_{in})$ that exactly divides $d_{in}$ exists, we have $d'_{out} = d_{in}$.
		Then, it needs $d'_{out}$ homomorphic multiplications and $\lfloor \log(d_{in} / d'_{out}) \rfloor + d_{in} / d'_{out} - 2^{\lfloor \log(d_{in} / d'_{out}) \rfloor}$ homomorphic rotations to evaluate $AB$, which corresponds to $O(d_{in})$ homomorphic multiplications and $O(d_{in}/d_{out})$ homomorphic rotations. 
		
		When $d_{in} > \lfloor \sqrt{N} \rfloor \geq d_{out}$, Matrix Squaring, i.e., the SOTA method with our extension, splits $A$ and $B$ to derive $\lceil d_{in} / \lfloor \sqrt{N} \rfloor \rceil$ pairs of $d_{out} \times \lfloor \sqrt{N} \rfloor$ matrices and $\lfloor \sqrt{N} \rfloor \times m$ matrices and applies the SOTA method to each pair, which in total results in $O(\sqrt{N}) \cdot \lceil d_{in} / \lfloor \sqrt{N} \rfloor \rceil = O(d_{in})$ homomorphic multiplications and $O(\sqrt{N} / d_{out}) \cdot \lceil d_{in} / \lfloor \sqrt{N} \rfloor \rceil = O(d_{in} / d_{out})$ homomorphic rotations.
		
		When $d_{in}  \geq d_{out} > \lfloor \sqrt{N} \rfloor$, Matrix Squaring splits $A$ and $B$ to derive $\lceil d_{in} / \lfloor \sqrt{N} \rfloor \rceil \cdot \lfloor d_{out} / \lfloor \sqrt{N} \rfloor \rfloor$ pairs of $\lfloor \sqrt{N} \rfloor \times \lfloor \sqrt{N} \rfloor$ matrices and $\lfloor \sqrt{N} \rfloor \times m$ matrices as well as $\lceil d_{in} / \lfloor \sqrt{N} \rfloor \rceil$ pairs of $(d_{out} \bmod \lfloor \sqrt{N} \rfloor) \times \lfloor \sqrt{N} \rfloor$ matrices and $\lfloor \sqrt{N} \rfloor \times m$ matrices, which totally needs $O(\sqrt{N}) \cdot \lceil d_{in} / \lfloor \sqrt{N} \rfloor \rceil \cdot \lceil d_{out} / \lfloor \sqrt{N} \rfloor \rceil = O(d_{in} \cdot d_{out} / \sqrt{N})$ homomorphic multiplications and $O(\sqrt{N} / (d_{out} \bmod \lfloor \sqrt{N} \rfloor)) \cdot \lceil d_{in} / \lfloor \sqrt{N} \rfloor \rceil = O(d_{in} / (d_{out} \bmod \sqrt{N}))$ homomorphic rotations.
	\end{proof}
	
	
	\begin{proof}[Proof of Lemma \ref{lemma:time_cost}]
		To test a single model on a single sample, HESV and SecSV need to perform $L$ matrix multiplications for the $L$ linear layers.
		These matrix multiplications result in $O(\sum_{l=1}^{L} (d^{(l)}_{in} \cdot d^{(l)}_{out} / \sqrt{N}))$ HMults and $O(\sum_{l=1}^{L} \frac{d_{in}}{d_{out} \bmod \sqrt{N}})$ HRots for HESV while $O(\sum_{l=1}^{L} d^{(l)}_{in})$ $HMult$s and no HRot for SecSV.
		
		Then, when testing a batch of test samples using the SIMD property of HE, we should decide on a batch size $m$ that is compatible with all the layers since they may have different shapes.
		Therefore, we have $m \leq \min\{d^{(1)}_{in},...,d^{(L)}_{in},\sqrt{N}\}$ for HESV and $m \leq \lfloor N/\max\{d^{(1)}_{out},...,d^{(L)}_{out} \} \rfloor$ for SecSV.
		
		Therefore, to test a single model on $M$ test samples, HESV needs $O( \frac{M \cdot (\sum_{l=1}^{L} d^{(l)}_{in} \cdot d^{(l)}_{out})}{\min\{d^{(1)}_{in},...,d^{(L)}_{in},\sqrt{N}\} \cdot \sqrt{N}})$ HMults and $O(\frac{M \cdot \sum_{l=1}^{L} \frac{d_{in}}{d_{out} \bmod \sqrt{N}}}{\min\{d^{(1)}_{in},...,d^{(L)}_{in},\sqrt{N}\}})$ HRots, while SecSV needs $O( \frac{M \cdot \max\{d^{(1)}_{out},...,d^{(L)}_{out}\} \sum_{l=1}^{L} d^{(l)}_{in}}{N})$ HMults and no HRot.
		Finally, since we have $O(2^{n})$ models to test for each training round $t \in [T]$, we conclude that Lemma \ref{lemma:time_cost} stands true.
	\end{proof}

	\begin{proof}[Proof of Theorem \ref{thm:sampleskip}]
		For any model $\theta\in [\theta_{S_1}, \theta_{S_2}]$, if $y=\argmax_{j} f_{\theta}(x)[j]$, we have $f_{\theta}(x)[y] > f_{\theta}(x)[j], \forall j \neq y$.
		Because $\eta$ is an entrywise strictly increasing function, we have
		\begin{align*}
			&\forall \theta \in [\theta_{S_1}, \theta_{S_2}], (\theta^w x + \theta^b)[y] > (\theta^w x + \theta^b)[j], \forall j \neq y \\
			&\Rightarrow (\theta_{\{{S_1}, {S_2}\}}^w x + \theta_{\{{S_1}, {S_2}\}}^b)[y] > (\theta_{\{{S_1}, {S_2}\}}^w x + \theta_{\{{S_1}, {S_2}\}}^b)[j], \forall j \neq y\\
			&\Rightarrow f_{\theta_{\{{S_1}, {S_2}\}}}(x)[y] > f_{\theta_{\{{S_1}, {S_2}\}}}(x)[j], \forall j \neq y
		\end{align*}
		from which we conclude that $\argmax_{j} f_{\theta_{\{{S_1}, {S_2}\}}}(x)[j] = y$.
	\end{proof}
	
	\begin{proof}[Proof of Lemma \ref{lemma:correctness}]
		For each round $t$, for each client $i \in I^t$, , we have
		\begin{align*}
			|\hat{\phi}^t_i - \phi^t_i| = & \big|\sum_{S\subseteq I^t \backslash\{i\}} \frac{|S|!(|I^t|-|S|-1)!}{|I^t|!}\big(\hat{v}(\theta^t_{S\cup \{i\}}) \!- \!\hat{v}(\theta^t_{S})\big) \\
			& - \sum_{S\subseteq I^t \backslash\{i\}} \frac{|S|!(|I^t|-|S|-1)!}{|I^t|!}\big(v(\theta^t_{S\cup \{i\}}) \!- \!v(\theta^t_{S})\big) \big|\\
			= & \big|\sum_{S\subseteq I^t \backslash\{i\}} \frac{|S|!(|I^t|-|S|-1)!}{|I^t|!} \big( (\hat{v}(\theta^t_{S\cup \{i\}}) - v(\theta^t_{S\cup \{i\}})) \\
			& - (\hat{v}(\theta^t_{S}) - v(\theta^t_{S})) \big) \big| \\
			= & \big|\sum_{S\subseteq I^t \backslash\{i\}} \frac{|S|!(|I^t|-|S|-1)!}{|I^t|}  (\Delta v(\theta^t_{S\cup \{i\}}) - \Delta v(\theta^t_{S})) \big| \\
			\leq & \sum_{S\subseteq I^t \backslash\{i\}} \frac{|S|!(|I^t|-|S|-1)!}{|I^t|} \big| (\Delta v(\theta^t_{S\cup \{i\}}) - \Delta v(\theta^t_{S})) \big| \\
			\leq & \sum_{S\subseteq I^t \backslash\{i\}} \frac{|S|!(|I^t|-|S|-1)!}{|I^t|}  \big| \Delta v_{max} \big| = \Delta v_{max}
		\end{align*}
		Therefore, for each client $i \in I$, we have
		$$|\hat{\phi}_i - \phi_i| = |\sum_{t=1}^T (\hat{\phi}^t_i - \phi^t_i)| \leq \sum_{t=1}^T |\hat{\phi}^t_i - \phi^t_i| \leq \sum_{t=1}^T \Delta v_{max} = T \cdot \Delta v_{max}$$
		
	\end{proof}
	
	\begin{proof}[Proof of Proposition \ref{prop:security}]
		Under HESV (SecSV), the server(s) can only obtain encrypted models, encrypted (secret shares of) input features, encrypted output features, encrypted (secret shares of) labels, and the count of (the indices of) correctly predicted samples.
		Since the probabilities of the server(s) successfully inferring input features, output features, or model parameters from the count of (the indices of) correctly predicted samples are negligible, according to Definitions \ref{def:es_attack}, \ref{def:mi_attack}, and \ref{def:retrain_attack}, the server(s) cannot apply the defined attacks due to the lack of the necessary messages.
		
		Let $i_{\theta}$ denote the owner of model $\theta$.
		For each layer $l$, if $n \geq 3$, given that client $i_{l}$ possseses the input features $X^{(l)}$, HESV and SecSV ensures that the output features $\hat{Y}^{(l)}$ are allocated to another client $i_{l+1} \neq i_{l} \neq i_{\theta}$.
		Since for all PPT adversaries $\mathcal{A}$ and for all linear layers $l$, there is a negligible $\epsilon$ such that $Pr[\mathcal{A}(\hat{Y}^{(l)}) = X^{(l)}] \leq \epsilon$ over all possible $(\theta^{(l)}, X^{(l)})$, we have
		\begin{align*}
			& Pr[(\mathcal{A}(\hat{Y}^{(l)}) = X^{(l)}) \wedge (\mathcal{A}(X^{(l)},\hat{Y}^{(l)}) = \theta^{(l)})] \\
			= & Pr[\mathcal{A}(X^{(l)}, \hat{Y}^{(l)}) \!=\! \theta^{(l)} | \mathcal{A}(\hat{Y}^{(l)}) \!=\! X^{(l)}]  \!\cdot\! Pr[\mathcal{A}(\hat{Y}^{(l)}) \!=\! X^{(l)}] \leq \epsilon
		\end{align*}
		Therefore, the probability of client $i_{l+1}$ applying single-layer ES attacks to infer $\theta^{(l)}$ is negligible.
		Also, since for all PPT adversaries $\mathcal{A}$ and for all linear layers $l$, there is a negligible $\epsilon$ such that $Pr[\mathcal{A}(\hat{Y}^{(l)}) = \theta^{(l)}] \leq \epsilon$ over all possible $(\theta^{(l)}, X^{(l)})$, we have
		\begin{align*}
			\medmath{Pr[\wedge_{l=1}^{l_2}\big((\mathcal{A}(\hat{Y}^{(l)})\! = \!\theta^{(l)}) \!\wedge \!(\mathcal{A}(\hat{Y}^{(l)}) \!=\! X^{(l)})\big) \wedge (\mathcal{A}(\{\theta^{(l)}\}_{l=1}^{l_2}, \hat{Y}^{(l_2)}) \!= \!X^{(1)})] \leq \epsilon}
		\end{align*}
		Therefore, the probability of applying ES data stealing attacks to infer the features $X^{(1)} = X$ of test data is negligible.
		
		Similarly, if $n \geq 4$, HESV and SecSV can ensure that client $i_{L+1}$ who receives the output features $\hat{Y}^{(L)}$ is not the owner $i_{1}$ of the input $X^{(1)} = X$.
		Then, the probability of client $i_{L+1}$ applying retraining attacks is negligible:
		\begin{align*}
			& Pr[\wedge_{l=1}^{L}(\mathcal{A}(\hat{Y}^{(l)}) = X^{(l)}) \wedge (\mathcal{A}(X^{(1)}, \hat{Y}) = f')]  \\
			= & Pr[\wedge_{l=1}^{L-1}(\mathcal{A}(\hat{Y}^{(l)}) = X^{(l)}) \wedge (\mathcal{A}(X^{(1)}, \hat{Y}) = f') | \mathcal{A}(\hat{Y}^{(L)}) = X^{(L)}]  \\
			& \cdot Pr[\mathcal{A}(\hat{Y}^{(L)}) = X^{(L)}] \leq \epsilon
		\end{align*}
		For each sample $(x,y)$ in the batch, assume that it was in the training set $\mathcal{D}^{train}$ without loss of generality;
		the probability of client $i_{L+1}$ applying an MI attack is negligible:
		\begin{align*}
			& Pr[\wedge_{l=1}^{L}(\mathcal{A}(\hat{Y}^{(l)}) = X^{(l)}) \wedge (\mathcal{A}((x,y), \hat{Y}) = ((x,y) \text{ was in } \mathcal{D}^{train}))]  \\
			= & \medmath{Pr[\wedge_{l=1}^{L-1}(\mathcal{A}(\hat{Y}^{(l)}) \!=\! X^{(l)}) \!\wedge\! (\mathcal{A}((x,y), \hat{Y}) \!=\! ((x,y) \text{ was in } \mathcal{D}^{train})) | \mathcal{A}(\hat{Y}^{(L)}) \!=\! X^{(L)}]}  \\
			&\cdot Pr[\mathcal{A}(\hat{Y}^{(L)}) = X^{(L)}] \leq \epsilon
		\end{align*}
		
		Finally, if $n \geq L + 2$, HESV and SecSV can ensure that all layers of input and output features are possessed by different clients and that the model owner receives nothing.
		In this case, for all possible $l_1, l_2$ such that $l_1 < l_2$, the probability of client $i_{l_2+1}$ applying multi-layer ES attacks to steal $\theta^{(l_1)},...,\theta^{(l_2)}$ is negligible: 
		\begin{align*}
			& Pr[\wedge_{l=l_1}^{l_2}(\mathcal{A}(\hat{Y}^{(l)}) = X^{(l)})  \wedge (\mathcal{A}(X^{(l_1)}, \hat{Y}^{(l_2)}) = (\theta^{(l_1)},...,\theta^{(l_2)}))]  \\
			= & \medmath{Pr[\wedge_{l=l_1}^{l_2 - 1}(\mathcal{A}(\hat{Y}^{(l)}) \!=\! X^{(l)})  \!\wedge\! (\mathcal{A}(X^{(l_1)}, \hat{Y}^{(l_2)}) \!=\! (\theta^{(l_1)},...,\theta^{(l_2)})) | \mathcal{A}(\hat{Y}^{(l_2)}) \!=\! X^{(l_2)}]} \\
			&\cdot Pr[\mathcal{A}(\hat{Y}^{(l_2)}) = X^{(l_2)}] \leq \epsilon
		\end{align*}
		
	\end{proof}

	\begin{table}[ht]
		\caption{Logistic classifier for AGNEWS and BANK.}
		\label{tab:logistic}
		\begin{tabular}{|l|l|}
			\hline
			FC      & \begin{tabular}[c]{@{}l@{}}Fully connects $d$ input features \\ to $c$ output features: $\mathbb{R}^{2\times 1}\leftarrow \mathbb{R}^{2\times 48}\times \mathbb{R}^{48 \times 1}$\end{tabular} \\ \hline
			Sigmoid & \begin{tabular}[c]{@{}l@{}}Calculates the sigmoid function \\ for each output feature\end{tabular}                                                                                             \\ \hline
		\end{tabular}
	\end{table}
	
	\begin{table}[ht]
		\caption{CNN classifier for MNIST.}
		\label{tab:minst_cnn}
		\begin{tabular}{|l|l|}
			\hline
			Conv.  & \begin{tabular}[c]{@{}l@{}}Input image $28\times28$, kernel size $7\times7$, strize size $3$, \\ output channels 4: $\mathbb{R}^{4\times64}\leftarrow \mathbb{R}^{4\times 28}\times \mathbb{R}^{28\times 64}$\end{tabular} \\ \hline
			Square & Squares $256$ output features                                                                                                                                                                                              \\ \hline
			FC     & \begin{tabular}[c]{@{}l@{}}Fully connects $256$ input features \\ to $64$ output features: $\mathbb{R}^{64\times 1}\leftarrow \mathbb{R}^{64\times 256}\times \mathbb{R}^{256 \times 1}$\end{tabular}                      \\ \hline
			Square & Squares $64$ output feautres                                                                                                                                                                                               \\ \hline
			FC     & \begin{tabular}[c]{@{}l@{}}Fully connects $64$ input features \\ to $10$ output features: $\mathbb{R}^{10\times 1}\leftarrow \mathbb{R}^{10\times 64}\times \mathbb{R}^{64 \times 1}$\end{tabular}                         \\ \hline
		\end{tabular}
	\end{table}
	
	\begin{table}[]
		\small
		\caption{RNN classifier for miRNA-mRNA.}
		\label{tab:mrna_rnn}
		\begin{tabular}{|l|l|}
			\hline
			\multirow{5}{*}{GRU} & Repeat the following steps $10$ times:                                                                                                                                                                                                                                                                                                                                                                                                                                                                                                                                                                                                                                                                                                                                                                                                                                             \\ \cline{2-2} 
			& \begin{tabular}[c]{@{}l@{}}(1) Fully connects $64$ input features to $32$ input-reset features, \\ fully connects $32$ hidden features to $32$ hidden-reset features,\\ aggregates input-reset features and hidden-reset features,\\ and calculates the sigmoid function of the aggregated result \\ as reset gates:\\ $\mathbb{R}^{32\times 1} \leftarrow \mathbb{R}^{32 \times 64} \times \mathbb{R}^{64\times1}$,\\ $\mathbb{R}^{32\times 1} \leftarrow \mathbb{R}^{32 \times 32} \times \mathbb{R}^{32\times1}$,\\ $\mathbb{R}^{32\times 1} \leftarrow \mathbb{R}^{32 \times 1} + \mathbb{R}^{32 \times 1}$,\\ $\mathbb{R}^{32\times 1} \leftarrow Sigmoid(\mathbb{R}^{32\times 1})$.\end{tabular}                                                                                                                                                                             \\
			& \begin{tabular}[c]{@{}l@{}}(2) Fully connects $64$ input features to $32$ input-update features, \\ fully connects $32$ hidden features to $32$ hidden-update features,\\ aggregates input-update features and hidden-update features,\\ and calculates the sigmoid function of the aggregated result \\ as update gates:\\ $\mathbb{R}^{32\times 1} \leftarrow \mathbb{R}^{32 \times 64} \times \mathbb{R}^{64\times1}$,\\ $\mathbb{R}^{32\times 1} \leftarrow \mathbb{R}^{32 \times 32} \times \mathbb{R}^{32\times1}$,\\ $\mathbb{R}^{32\times 1} \leftarrow \mathbb{R}^{32 \times 1} + \mathbb{R}^{32 \times 1}$,\\ $\mathbb{R}^{32\times 1} \leftarrow sigmoid(\mathbb{R}^{32\times 1})$.\end{tabular}                                                                                                                                                                        \\
			& \begin{tabular}[c]{@{}l@{}}(3) Fully connects $64$ input features to $32$ input-new features, \\ fully connects $32$ hidden features to $32$ hidden-new features,\\ multiplies $32$ reset gates with $32$ hidden-new features\\ to derive $32$ reset features,\\ and aggregates input-new features and reset features,\\ and calculates the tanh function of the aggregated result \\ as new gates:\\ $\mathbb{R}^{32\times 1} \leftarrow \mathbb{R}^{32 \times 64} \times \mathbb{R}^{64\times1}$,\\ $\mathbb{R}^{32\times 1} \leftarrow \mathbb{R}^{32 \times 32} \times \mathbb{R}^{32\times1}$,\\ $\mathbb{R}^{32\times 1} \leftarrow \mathbb{R}^{32 \times 1} \cdot \mathbb{R}^{32 \times 1}$,\\ $\mathbb{R}^{32\times 1} \leftarrow \mathbb{R}^{32 \times 1} + \mathbb{R}^{32 \times 1}$,\\ $\mathbb{R}^{32\times 1} \leftarrow tanh(\mathbb{R}^{32\times 1})$.\end{tabular} \\
			& \begin{tabular}[c]{@{}l@{}}(4) Aggregate $32$ hidden features and $32$ new gates \\ with weights of update gates and (1 - update gates)\\ to update hidden features:\\ $\mathbb{R}^{32\times 1} \leftarrow \mathbb{R}^{32 \times 1} \cdot \mathbb{R}^{32\times1}$,\\ $\mathbb{R}^{32\times 1} \leftarrow \mathbb{R}^{32 \times 1} \cdot \mathbb{R}^{32\times1}$,\\ $\mathbb{R}^{32\times 1} \leftarrow \mathbb{R}^{32 \times 1} + \mathbb{R}^{32 \times 1}$.\end{tabular}                                                                                                                                                                                                                                                                                                                                                                                                          \\ \hline
			FC                   & Fully connects $32$ input features to $2$ output features.                                                                                                                                                                                                                                                                                                                                                                                                                                                                                                                                                                                                                                                                                                                                                                                                                         \\ \hline
		\end{tabular}
	\end{table}
	
	\begin{table}[ht]
		\caption{DNN classifier.}
		\label{tab:dnn}
		\begin{tabular}{|l|l|}
			\hline
			\begin{tabular}[c]{@{}l@{}}Layer $1$\\ (input)\end{tabular}               & \begin{tabular}[c]{@{}l@{}}Fully connects $d$ input features\\ to $100$ output features:\\ $\mathbb{R}^{100\times1} \leftarrow \mathbb{R}^{100\times d}\times \mathbb{R}^{d\times 1}$\end{tabular}     \\ \hline
			\begin{tabular}[c]{@{}l@{}}Layer $l\in \{2,L-1\}$\\ (hidden)\end{tabular} & \begin{tabular}[c]{@{}l@{}}Fully connects $100$ input features\\ to $100$ output features:\\ $\mathbb{R}^{100\times1}\leftarrow \mathbb{R}^{100\times100}\times \mathbb{R}^{100\times1}$.\end{tabular} \\ \hline
			\begin{tabular}[c]{@{}l@{}}Layer $L$\\ (output)\end{tabular}              & \begin{tabular}[c]{@{}l@{}}Fully connects $100$ input features\\ to $c$ output features:\\ $\mathbb{R}^{c\times1} \leftarrow \mathbb{R}^{c\times 100}\times \mathbb{R}^{100\times 1}$\end{tabular}     \\ \hline
		\end{tabular}
	\end{table}
	
	\section{Datasets and Classifiers}
	\label{appendix:data_model}
	\subsubsection*{Datasets}
	We use the following datasets for experiments.
	\begin{itemize}[leftmargin=*]
		\item \textbf{AGNEWS \cite{zhang2015character}}: It is a subdataset of AG's corpus news articles for news classification consisting of $120000$ training samples and $7600$ test samples with $4$ classes of labels.
		\item \textbf{BANK \cite{moro2014data}}: This dataset contains $41188$ samples of bank customers' personal information for bank marketing.
		The goal is to predict whether a customer will subscribe a term deposit or not, which is a binary classification task.
		We split the dataset into $31188$ samples for training and $10000$ samples for testing.
		\item \textbf{MNIST \cite{lecun1998gradient}}: It is a dataset of handwritten digits and is one of the most famous benchmark for image classification.
		It has $60000$ training images and $10000$ test images, and the task is to identify the digit from $0$ to $9$ in a given image.
		\item \textbf{miRNA-mRNA \cite{menor2014mirmark}}: This dataset comprises human miRNA-mRNA pairs for miRNA target prediction, with $62642$ pairs for training and $3312$ pairs for testing.
		The task is to predict whether a given mRNA is the target of a given miRNA.
	\end{itemize}
	
	\subsubsection*{Classifiers}
	We use the following classifiers for experiments.
	\begin{itemize}[leftmargin=*]
		\item \textbf{Logistic}: For the BANK and AGNEWS datasets, we use the logistic classifier with $1$ fully-connected (FC) layer and the sigmoid activation (see Table \ref{tab:logistic}).
		\item \textbf{CNN}: For the MNIST dataset, we follow \cite{jiang2018e2dm} to use a convolutional NN (CNN) with $1$ convolution layer and $2$ FC layers with the square activation (see Table \ref{tab:minst_cnn}).
		\item \textbf{RNN}: For the miRNA-mRNA dataset, we adopt a recurrent NN named deepTarget \cite{lee2016deeptarget} with $1$ gate recurrent units (GRUs) layer and $1$ FC layer (see Table \ref{tab:mrna_rnn}).
		\item \textbf{DNN}: For all the datasets, we adopt the widely-used DNN classifier with $1$ FC input layer, $1$ FC output layer, and $L-2$ FC hidden layers with $100$ hidden nodes (see Table \ref{tab:dnn}).  
\end{itemize}}
	
\end{document}